%% file: cut_balance_for_arXiv.tex
\documentclass[journal]{IEEEtran}

\usepackage{amsthm}

\include{declarations}
\usepackage{color}

\def\jnt#1{{#1}}
\def\green#1{{#1}}
\def\TN#1{{#1}}
\def\rrev#1{{#1}}

\title{Convergence of type-symmetric and cut-balanced  
consensus seeking systems (extended version)\thanks{This is an extended version of 
\cite{HendrickxTsitsiklis_cutbalanced:2011_cdc} and \cite{HendrickxTsitsiklis_cutbalanced:2013}; it  includes proofs and some additional results that were omitted from these papers.
This research was supported by the National Science Foundation
under grant ECCS-0701623, by the Concerted Research Action (ARC) \quotes{Large Graphs and Networks} of the French Community of Belgium, by the Belgian Programme on Interuniversity Attraction Poles initiated by the Belgian Federal Science Policy Office, and by postdoctoral fellowships from
the F.R.S.-FNRS (Belgian Fund for Scientific Research) and the
B.A.E.F. (Belgian American Education Foundation).}}

\author{Julien M. Hendrickx\thanks{Universit\'e
catholique de Louvain, B-1348 Louvain-la-Neuve, Belgium; {\tt \small julien.hendrickx@uclouvain.be.}
This research was conducted while visiting LIDS, MIT.} \ 
and John N.\ Tsitsiklis\thanks{Laboratory
for Information and Decision Systems, Massachusetts Institute of
Technology, Cambridge, MA 02139, USA; {\tt \small
jnt@mit.edu}.}
}

\begin{document}
\maketitle \thispagestyle{empty}

\input{abstract}

\input{intro}
\input{proof}

\input{part_cases}
\input{random}
\input{endogenous}
\input{discrete-time}
\input{conclusions}

\input{biblio}
\appendix
\input{appendix}

\end{document}

%% file: declarations.tex
\usepackage{amssymb,amsmath}
\usepackage{amsfonts}
\usepackage[english]{babel}
\usepackage[latin1]{inputenc}
\usepackage{fancyhdr}
\usepackage{yfonts}
\usepackage{mathrsfs}
\usepackage[dvips]{graphicx}
\usepackage[rflt]{floatflt}
\usepackage{cite}

\usepackage{algorithm}
\floatname{algorithm}{Alg.}
\usepackage{algorithmic}
\usepackage{psfrag}
\usepackage{url}

\providecommand{\com[2]}{\begin{tt}[#1: #2]\end{tt}}
\usepackage{color}

\newtheorem{prop}{Proposition}
\newtheorem{thm}{Theorem}

\newtheorem{cor}{Corollary}

\newtheorem{lem}{Lemma}
\newtheorem{assum}{Assumption}
\newtheorem{exmp}{Example}

\providecommand{\prt}[1]{\left( #1 \right)}
 
\providecommand{\e}{\phantom{aaa}}
\providecommand{\abs[1]}{\left|#1\right|}

\providecommand{\quotes[1]}{``#1"}

%% file: abstract.tex
\begin{abstract}
We consider continuous-time consensus seeking systems whose time-dependent interactions are cut-balanced, in the following sense:  if a group of agents influences the remaining ones, the former group is also influenced by the remaining ones by at least a proportional amount. Models involving symmetric interconnections and models in which a weighted average of the agent values is conserved are special cases. We prove that such systems always converge. We give a sufficient condition on the evolving interaction topology for the limit values of two agents to be the same. Conversely, we show that if our condition is not satisfied, then these limits are generically different. These results allow treating systems where the agent interactions are a priori unknown, e.g., random or determined 
endogenously by the agent values. We also derive corresponding results for discrete-time systems. 
\end{abstract}

%% file: intro.tex
%%%%%%%%%%%%%%%%%%%%%%%%%%%%%%%%%%%%%%%%%%%%%%%%%%%%%%%%%%%%%%%%%
%%%%%%%%%%%%%%%%%%%%%%%%%%%%%%%%%%%%%%%%%%%%%%%%%%%%%%%%%%%%%%%%%
%%%%%%%%%%%%%%%%%%%%%%%%%%%%%%%%%%%%%%%%%%%%%%%%%%%%%%%%%%%%%%%%%
\section{Introduction}
%%%%%%%%%%%%%%%%%%%%%%%%%%%%%%%%%%%%%%%%%%%%%%%%%%%%%%%%%%%%%%%%%
%%%%%%%%%%%%%%%%%%%%%%%%%%%%%%%%%%%%%%%%%%%%%%%%%%%%%%%%%%%%%%%%%

We consider continuous-time consensus seeking systems of the following kind: each of $n$ agents, indexed by $i=1,\ldots,n$, maintains a value $x_i(t)$, which is a continuous function of time and evolves according to the integral equation version of
\begin{equation}\label{eq:system_diff}
\frac{d}{dt}x_i(t) = \sum_{j=1}^n a_{ij}(t) \prt{x_j(t)-x_i(t)}.
\end{equation}
Throughout we assume that each $a_{ij}(\cdot)$ is a \emph{nonnegative} and measurable function. 
We introduce the following assumption which plays a central role in this paper.

\begin{assum}\label{as:cb} (Cut-balance)
There exists a constant $K\geq 1$ such that for all $t$, and any nonempty proper subset $S$ of $\{1,\dots,n\}$, we have\footnote{Note that the second inequality, added to emphasize the symmetry of the condition, is redundant.}
\begin{equation}\label{eq:cut_balance_intro}
 K^{-1} \sum_{i\in S,j\notin S} a_{ji}(t) \leq \sum_{i\in
S,j\notin S} a_{ij}(t) \leq K\sum_{i\in S,j\notin S} a_{ji}(t).
\end{equation}
\end{assum}

Intuitively, if a group of agents influences the remaining ones, the former 
group is also influenced by the remaining ones by at least a proportional amount. This condition may seem hard to verify in general. But, several important particular classes of consensus seeking systems automatically satisfy it. These include
symmetric systems ($a_{ij}(t) = a_{ji}(t)$), type-symmetric systems ($a_{ij}(t) \leq Ka_{ji}(t)$), and, as will be seen later, any system whose dynamics conserve a weighted average (with positive weights) of the agent values.

Under the cut-balance condition (\ref{eq:cut_balance_intro}), 
and without any further assumptions, we prove that each value $x_i$ converges to a limit.
Moreover, we show that $x_i$ and $x_j$ converge to the same limit
if $i$ and $j$ belong to the same connected component of the ``unbounded interactions graph,'' i.e., the graph whose edges correspond to the pairs $(j,i)$ for which
$\int_0^\infty a_{ij}(\tau)\, d\tau$ is unbounded. (As we will show, while this is a directed graph,
each of its weakly connected components is also strongly
connected.) Conversely, we prove that $x_i$ and $x_j$ generically converge to
different limits if $i$ and $j$ belong to different connected components
of that graph. (This latter result involves an additional technical  assumption that $\int_0^Ta_{ij}(\tau)\,d\tau<\infty$  for every $T<\infty$.) 

Our method of proof is different from traditional convergence proofs
for consensus seeking systems, which rely on either span-norm or quadratic norm contraction properties. 
It consists of showing that for every $m\leq n$, a particular
linear combination of the values of 
the $m$ agents with the smallest values
is nondecreasing and bounded, and that its total increase rate  eventually becomes bounded below 
by a positive number if two agents with
unbounded interactions were to converge to different limits. The idea of working with these linear combinations is inspired from and extends a technique used in
\cite{BlondelHendrickxTsitsiklis:2009_APCT} to analyze a particular average-preserving system. More specifically, \cite{BlondelHendrickxTsitsiklis:2009_APCT} analyses in depth a  \TN{model of opinion dynamics for which the order between the agents is preserved, where the coefficients switch between $0$ and $1$, are symmetric ($a_{ij}=a_{ji}$), and switch at most finitely often in any finite interval.} Convergence  is obtained by proving that the average value of the $m$ first agents is nondecreasing and bounded, for any $m$.

Motivation for our model comes from the fact that there are many systems 
in which an agent cannot influence the others without
being subjected to at least a fraction of the reverse influence.
This is, for example, a common assumption in numerous models of social
interactions and opinion dynamics \cite{Lorenz:2007, CastellanoFortunatoLoreto:2009},  or physical systems.

%%%%%%%%%%%%%%%%%%%%%%%%%%%%%%%%%%%%%%%%%%%%%%%%%%%%%%%%%%%%%%%%%
\subsection{Background}
%%%%%%%%%%%%%%%%%%%%%%%%%%%%%%%%%%%%%%%%%%%%%%%%%%%%%%%%%%%%%%%%%

Systems of the form (\ref{eq:system_diff}) have attracted 
considerable attention \cite{Moreau:2005,OlfatiSaberMurray:2004,XiaoWang:2008,LinBrouckeFrancis:2004,RenBeard:2005}
(see also 
\cite{RenBeardAtkins:2007,OlfatiSaberFaxMurray:2006} for surveys), 
with motivation coming from 
decentralized
coordination,  data fusion \cite{XiaoBoydLall:2005,BoydGhoshPrabhakarShah:2005}, 
animal flocking 
\cite{JadbabaieLinMorse:2003,VicsekCzirolBenjacobCohenSchchet:1995,Chazelle:2009},
and models of social behavior \cite{Lorenz:2007,CastellanoFortunatoLoreto:2009,Lorenz:2005,BlondelHendrickxTsitsiklis:2009_APCT, BlondelHendrickxTsitsiklis:2009_Krausemodel,HegselmannKrause:2002,Ben-naimKrapivskyRedner:2003}.

Available results impose some connectivity conditions 
on the evolution of the coefficients $a_{ij}(t)$, and 
usually guarantee exponentially fast convergence of each agent's value to a common limit (``consensus'').  
For example, Olfati-Saber and Murray
\cite{OlfatiSaberMurray:2004} consider the system
$$
\frac{d}{dt}x_i(t) = \sum_{j:(j,i)\in E(t)}\prt{x_j(t)-x_i(t)}
$$
with a time-varying directed graph $G(t)=(\{1,\dots,n\},E(t))$; this is a special case of the model  
(\ref{eq:system_diff}), with $a_{ij}(t)$ equal to one 
if $(j,i)\in E(t)$, and equal to zero otherwise. 
They show that if  the out-degree of every node is equal to its
in-degree at all times, and if each graph $G(t)$ is
strongly connected, then 
the system is average-preserving and 
each $x_i(t)$ converges exponentially fast to
$\frac{1}{n}\sum_{j}x_j(0)$. They also obtain similar results for systems with arbitrary but fixed coefficients $a_{ij}$. Moreau \cite{Moreau:2004} establishes
exponential convergence to consensus under weaker conditions: he only assumes that  
the $a_{ij}(t)$ are uniformly bounded, and that there exist 
$T>0$ and $\delta>0$ such that the directed graph obtained by
connecting $i$ to $j$ whenever $\int_{t}^{t+T}a_{ij}(\tau)\,d\tau>\delta$
has a rooted spanning tree, for every $t$. 
Several extensions of such
results, involving for example time delays or imperfect communications,
are also available.\footnote{It is common in the literature to treat
the system \eqref{eq:system_diff} as if the derivative existed for all $t$, which is not always the case. Nevertheless, such results remain correct under an appropriate reinterpretation of \eqref{eq:system_diff}.}

All of the above described results involve conditions that are 
easy to describe but difficult to ensure a priori, especially when the agent interactions are endogenously determined. This motivates the current work, which aims at an understanding of the convergence properties of the dynamical system \eqref{eq:system_diff} under minimal conditions. In the complete absence of any conditions, 
and especially in the absence of symmetry, 
it is well known that consensus seeking systems can fail to converge;
see e.g., Ch.\ 6 of \cite{BertsekasTsitsiklis:1989}.
On the other hand, it is also known that more predictable behavior and positive results are possible
in the following two cases: (i) symmetric (suitably defined) interactions, or (ii) average-preserving systems (e.g., in discrete-time models that involve  doubly stochastic matrices).

%%%%%%%%%%%%%%%%%%%%%%%%%%%%%%%%%%%%%%%%%%%%%%%%%%%%%%
%%%%%%%%%%%%%%%%%%%%%%%%%%%%%%%%%%%%%%%%%%%%%%%%%%%%%%
\subsection{Our contribution}
%%%%%%%%%%%%%%%%%%%%%%%%%%%%%%%%%%%%%%%%%%%%%%%%%%%%%%

Our cut-balance condition subsumes the two cases discussed above, and allows us to obtain strong convergence results. Indeed, we prove
convergence (not necessarily to consensus) without any additional condition, and then provide sufficient and
(generically) necessary conditions for the limit values of any two agents to agree. 
In contrast, existing  results show \emph{convergence to consensus} under some
fairly strong assumptions about persistent %repeated or permanent 
global connectivity, but 
offer no insight on the possible behavior when convergence to consensus
fails to hold.

The fact that our convergence result requires no assumptions other than the cut-balance condition is significant because it allows us to study 
systems for which
the evolution of $a_{ij}(t)$ is a priori unknown,  possibly random or dependent on the vector $x(t)$ itself.
In the latter type of models, with endogenously determined agent interconnections, 
it is essentially impossible to check a priori the connectivity conditions imposed in 
existing results, and such results are therefore inapplicable. 
In contrast, our results apply as long as the cut-balance condition
is satisfied. The advantage of this condition is that it can be often guaranteed a priori, e.g., if the system is naturally type-symmetric.

On the technical side, we note that similar convergence results are available for the special case of discrete-time symmetric or type-symmetric
systems 
\cite{Lorenz:2005,HendrickxBlondel:2006_MTNS,BlondelHendrickxOlshevskyTsitsiklis:2005,Hendrickx:2008phdthesis,LiWang:2004,Moreau:2005,Lorenz:2003diplomathesis},
though they are obtained with a different methodology. 
Discrete time is indeed much simpler because one can exploit the following fact:
either two agents 
interact on a
set of infinite length or they stop interacting after a
certain finite time. 
We will indeed show that such existing discrete-time results can be easily 
extended to the cut-balanced case.

%%%%%%%%%%%%%%%%%%%%%%%%%%%%%%%%%%%%%%%%%%%%%%%%%%%%%%
%%%%%%%%%%%%%%%%%%%%%%%%%%%%%%%%%%%%%%%%%%%%%%%%%%%%%%
\subsection{Outline}
%%%%%%%%%%%%%%%%%%%%%%%%%%%%%%%%%%%%%%%%%%%%%%%%%%%%%%

The remainder of the paper is organized as follows. We state and prove our main results in Section \ref{sec:main_result} and expose several particular classes of cut-balanced dynamics in Section \ref{sec:part_cases}. We demonstrate the application of our results to systems with randomly determined interactions in Section \ref{sec:random}, and to systems with endogenously determined interactions in Section \ref{sec:application_endogenous}. We show an analogous result for discrete-time systems in Section \ref{sec:discrete-time}. We end with some concluding remarks and the discussion of an open problem on the generalization of our results to systems involving a continuum of agents in Section \ref{sec:ccl}. The Appendix  contains the proof of a technical result needed in Section \ref{sec:main_result}.

%% file: proof.tex
%%%%%%%%%%%%%%%%%%%%%%%%%%%%%%%%%%%%%%%%%%%%%%%%%%%%%%%%%%%%%%
%%%%%%%%%%%%%%%%%%%%%%%%%%%%%%%%%%%%%%%%%%%%%%%%%%%%%%%%%%%%%%
%%%%%%%%%%%%%%%%%%%%%%%%%%%%%%%%%%%%%%%%%%%%%%%%%%%%%%%%%%%%%%
\section{Main Convergence Result and Proof}\label{sec:main_result}
%%%%%%%%%%%%%%%%%%%%%%%%%%%%%%%%%%%%%%%%%%%%%%%%%%%%%%%%%%%%%%
%%%%%%%%%%%%%%%%%%%%%%%%%%%%%%%%%%%%%%%%%%%%%%%%%%%%%%%%%%%%%%

We now state formally our main theorem, based on an
integral formulation of the agent dynamics.
The integral formulation avoids issues related
to the existence of derivatives, while allowing for discontinuous 
coefficients $a_{ij}(t)$ and 
possible Zeno behaviors (i.e., a countable number of discontinuities in a finite time interval).

Without loss of generality, we assume that $a_{ii}(t)=0$ for all $t$.
We define a directed  graph,  
$G=(\{1,\dots,n\},E)$, called the \emph{unbounded interactions graph}, by letting $(j,i)\in E$ if $\int_{0}^\infty
a_{ij}(t)\,dt = \infty$.

\begin{lem}
Suppose that Assumption \ref{as:cb} (cut-balance) holds. Every weakly connected component of the unbounded interaction graph $G$ is strongly connected. \TN{Equivalently, if there is a directed path from $i$ to $j$, then there is also a directed path from $j$ to $i$.}
\end{lem}
\begin{proof}
Consider a weakly connected component $W$ of the graph $G$. We assume, in order to derive a contradiction, that $W$ is not strongly connected. Consider the
decomposition of $W$ into strongly connected components. More precisely, we partition the nodes in $W$ into two or more subsets $C_1,C_2,\ldots$, so that each $C_k$ is strongly connected and so that any edge in $W$ that leaves a strongly connected component
leads to a component with a larger label: if $i\in C_k$, $(i,j)\in E$,
and $j\notin C_k$, then $j\in C_l$ for some $l>k$. (This decomposition is unique up to certain permutations of the subset labels.)
In particular, there is no edge from $C_2\cup C_3\cup\cdots$, that leads into $C_1$. Since $W$ is weakly connected, there must therefore exist an
edge from $C_1$ into $C_2\cup C_3\cup\cdots$. However, it is an immediate  consequence of the cut-balance condition (applied to $S=C_1$) that if there is an edge 
that leaves $C_1$ there must also exist an edge that enters $C_1$. This is a contradiction, and the proof is complete.
\end{proof}

The following assumption will be in effect in some of the results.
\begin{assum}\label{as:bi} (Boundedness) 
For every $i$ and $j$, and every $T<\infty$,
$\int_{0}^Ta_{ij}(t)\, dt<\infty$.
\end{assum}

We now state our main result.
\begin{thm}\label{thm:main_convergence}
Suppose that Assumption \ref{as:cb} (cut-balance) holds.
Let $x:\Re^+\to \Re^n$ be a solution to 
the system of integral equations
\begin{equation}\label{eq:evol_x}
x_i(t) = x_i(0) + \int_{0}^t \sum_{j=1}^n a_{ij}(\tau)\prt{x_j(\tau) -
x_i(\tau)}d\tau,
\end{equation}
for $i=1,\ldots,n$. Then:
\begin{enumerate}
\item[(a)]
The limit $x_i^* = \lim_{t\to \infty}x_i(t)$ exists, and\\ $x_i^*\in
[\min_jx_j(0),\max_jx_j(0)]$, for all $i$.
\item[(b)] For every $i$ and $j$, 
$\int_0^\infty a_{ij}(t)\abs{x_j(t)-x_i(t)}\,dt <\infty$. Furthermore, if
$i$ and $j$ belong to the same connected component of $G$,
then $x_i^* = x_j^*$. 
\end{enumerate}
If, in addition, Assumption \ref{as:bi} (boundedness) holds,  then:
\begin{enumerate}

\item[(c)] If $i$ and $j$ are in different connected components of $G$, then $x_i^*\neq x_j^*$, unless $x(0)$ belongs to a particular
$n-1$ dimensional sub-space of $\Re^n$, determined by the
functions $a_{ij}(\cdot)$.
\end{enumerate}
\end{thm}

The proof of convergence  relies on the fact that a particular
weighted sum of the $m$ smallest components $x_i(t)$ is nondecreasing. 
In
order to provide the intuition behind this proof, we first sketch 
the argument 
for the special case where: (i) the functions $x_i(\cdot)$ are differentiable, (ii)
the ordering of the components $x_i(t)$ changes at most a finite number of times during any
finite time-interval, and (iii) the coefficients satisfy
the stronger type-symmetry assumption: $0\leq K^{-1}a_{ji}(t) \leq
a_{ij}(t) \leq K a_{ji}(t).$

For any time $t$, let $y_1(t)$ be the smallest of the components $x_i(t)$, $y_2(t)$
the second smallest, etc., and let us define 
$S_m(t) = \sum_{i=1}^m K^{-i}y_i(t)$. Observe that the $y_i(\cdot)$ are continuous.
Indeed, any two components $x_i(\cdot)$ whose order is reversed at some time,
must be equal at that time, by continuity of $x(\cdot)$. As a
result, $S_m(\cdot)$ is also continuous.

Since we assume that the ordering of the $x_i(t)$ changes at most
a finite number of times during a finite time-interval, there exists an
increasing and divergent sequence $t_0,t_1,t_2,\dots$ such that the ordering  remains constant on any $(t_k,t_{k+1})$: It suffices to let the $t_k$ be the times at which the order changes, and to complete the sequence by an arbitrary diverging sequence if there are only finitely many order changes.
We consider such an
interval, on which we assume without loss of generality that the
$x_i(t)$ are already sorted, in nondecreasing order. Thus, for any $t\in (t_k,t_{k+1})$, and
$i>j$, we have $x_i(t)\geq x_j(t)$. Also, 
$y_i(t)=x_i(t)$,  and thus $S_m(t) = \sum _{i=1}^m K^{-i}\jnt{x}_i(t)$.
It follows from (\ref{eq:system_diff}) that, for $t\in
(t_k,t_{k+1})$,
\begin{align*}
\frac{d}{dt}S_m(t) &= \sum_{i=1}^m K^{-i}\frac{d}{dt} x_i(t)\\ &= \sum_{i=1}^m
K^{-i}\Big(\sum_{j=1}^n a_{ij}(t)\big(x_j(t)-x_i(t)\big)
\Big).
\end{align*}
This can be rewritten as
\begin{align*}
\frac{d}{dt}S_m(t) =& \sum_{i=1}^m \sum_{j=1}^m K^{-i} a_{ij}(t)
\prt{x_j(t)-x_i(t)} \\ +& \sum_{i=1}^m \sum_{j=m+1}^n K^{-i} a_{ij}(t)
\prt{x_j(t)-x_i(t)}.
\end{align*}
The second term on the right-hand side is always nonnegative, because the coefficients
$a_{ij}(t)$ are nonnegative and because we have $x_j(t)-x_i(t)\geq 0$ when $j>m\geq i$. By rearranging the first term, we
obtain 
\begin{equation}\label{eq:deriv_S_m geq}
\frac{d}{dt}S_m(t) \geq 
\sum_{i=1}^{m-1} \sum_{j=i}^m
\prt{K^{-i}a_{ij}(t)- K^{-j}a_{ji}(t) }\prt{x_j(t)-x_i(t)}.
\end{equation}
Recall the assumption $a_{ij}(t) \geq K^{-1}a_{ji}(t)$ for any $i,j$.
It implies that $$K^{-i}a_{ij}(t)- K^{-j}a_{ji}(t)\geq
a_{ji}(t)(K^{-(i+1)}-K^{-j}).$$ When $j> i$, as in
(\ref{eq:deriv_S_m geq}), this quantity is nonnegative, and so is
$x_j(t)-x_i(t)$. It then follows from (\ref{eq:deriv_S_m geq})
that $\frac{d}{dt}S_m(t) \geq 0$ for all $t\in (t_k,t_{k+1})$.

Since our argument is valid for any interval $(t_k,t_{k+1})$, and
since the sequence $t_0,t_1,\dots$ diverges, this implies that
$S_m(\cdot)$ is nondecreasing. Observe now that $x_i(t) \leq \max _j
x_j(0)$, for all $i$ and $t\geq 0$, because equation (\ref{eq:evol_x}) and the non-negativity of the $a_{ij}$ imply that $\max_i x_i(t)$ is nonincreasing. Thus, $S_m(t)$ is bounded above,
and therefore convergent, for any $m$. As a result, all $y_i(t)$
converge, i.e., the smallest entry of $x_i(t)$ converges, the
second smallest converges, etc. A continuity
argument can then be used to show the convergence of each $x_i(t)$.

The proof for the general case relies on the following lemma on the rate of change of the weighted sums $S_m(t)$, when the
coefficients $a_{ij}$ satisfy the cut-balance assumption. We say that a vector $y\in\Re^n$ is \emph{sorted} if $y_1\leq y_2\leq \dots
\leq y_n$. 

\begin{lem}\label{lem:deriv_weighted_avg_cutbal}
For $i,j=1,\ldots,n$, $i\neq j$, let $b_{ij}$ be nonnegative coefficients that satisfy the cut-balance condition
$$K^{-1}\sum_{i\in
S}\sum_{j\notin S}b_{ji} \leq \sum_{i\in S}\sum_{j\notin S}b_{ij}
\leq K\sum_{i\in S}\sum_{j\notin S}b_{ji},
$$
for some $K\geq 1$, and for every nonempty proper subset $S$ of $\{1,\ldots,n\}$.
Then,
$$
\sum_{i=1}^mK^{-i}\prt{\sum_{j=1}^n b_{ij}(y_j-y_i)}\geq 0
$$
for every sorted vector $y\in \Re^n$, and every $m\leq n$.
\end{lem}

\begin{proof}
We prove the stronger result that
\TN{$
\sum_{i=1}^{n}w_i\prt{\sum_{j=1}^n b_{ij}(y_j-y_i)}\geq 0
$}
for any nonnegative weights $w_i$ such that $w_i\geq K w_{i+1}$
for $i=1,\dots,n-1$. Observe that the expression above can be
rewritten as
$$
\sum_{i =1}^n y_i\prt{\sum_{j=1}^n
w_jb_{ji}-\sum_{j=1}^nw_ib_{ij}}=\sum_{i =1}^n y_i q_i,
$$
where the last equality serves as the definition of $q_i$.
We observe that
$$\sum_{i =1}^n y_i q_i = y_1\sum_{i=1}^n q_i
+\sum_{k=1}^{n-1}\Big( (y_{k+1}-y_{k})\sum_{i=k+1}^n q_i\Big).$$
It follows 
that the desired inequality $\sum_{i=1}^n y_i q_i \geq 0$
holds for every sorted vector $y$ 
if and only if (i) $\sum_{i=1}^n q_i= 0$,  and (ii) $\sum_{i=k+1}^n q_i\geq 0$, for 
$k=1,\dots, n-1$. We have 
\begin{align*}
\sum_{i =1}^n q_i&=
\sum_{i =1}^n \prt{\sum_{j=1}^n w_jb_{ji}-\sum_{j=1}^nw_ib_{ij}}\\&
= \sum_{i=1}^n\sum_{j=1}^nw_jb_{ji}-
\sum_{j=1}^n\sum_{i=1}^nw_ib_{ij} \\&=0,
\end{align*}
which establishes property (i). To establish property (ii), we observe that 
\begin{align*}
\sum_{i = k+1}^n  \Big(\sum_{j=1}^n
w_jb_{ji} &-\sum_{j=1}^n w_i b_{ij} \Big) = \\ &
\sum_{i=k+1}^n\sum_{j=1}^k w_j b_{ji} +
\sum_{i=k+1}^n\sum_{j=k+1}^n w_j b_{ji} \\  - &\sum_{i=k+1}^n\sum_{j=1}^k
w_i b_{ij} - \sum_{i=k+1}^n\sum_{j=k+1}^n w_i b_{ij}.
\end{align*}
The second and the fourth terms cancel each other. We use the inequality
$w_j\geq w_{k}$ for $j\leq k$ in the first term, and the inequality
$w_i\leq w_{k+1}$ for $i\geq k+1$ in the third term, to obtain
\begin{align*}
\sum_{i =k+1}^n q_i &=
\sum_{i = k+1}^n \prt{\sum_{j=1}^n
w_jb_{ji}-\sum_{j=1}^nw_ib_{ij}} \\ &\geq w_k
\sum_{i=k+1}^n\sum_{j=1}^kb_{ji} -
w_{k+1}\sum_{i=k+1}^n\sum_{j=1}^kb_{ij}.
\end{align*}
Using  the property $w_k\geq
Kw_{k+1}$, and the cut-balance assumption, we conclude 
that the right-hand side in the above inequality is nonnegative, which completes
the proof of property (ii).
We now let $w_i=K^{-i}$, for $i=1,\ldots,m$, and $w_i=0$ for $i>m$, to obtain the desired result.
\end{proof}

We now prove Theorem \ref{thm:main_convergence}.
%%%%%%%%%%%%%%%%%%%%%%%%%%%%%%%%%%%%%%%%%%%%%%%%%%%%%%%%%%%%%%%%%%%%%%%%%%%%%%%%%%%%%%%%
\begin{proof}[Proof (of Theorem 1)]
For every $t$, we define a permutation
$p(t)$ of the indices $\{1,\dots,n\}$ 
which sorts the components of the vector $x(t)$. 
(More precisely, it sorts the pairs $(x_i(t),i)$ in lexicographic order.) 
In particular, $p_i(t)$ is the index of the $i$th smallest component
of $x(t)$, with ties broken according to the original indices of the
components of $x(t)$.
Formally,  if $i <j$, then either $x_{p_i(t)}(t)< x_{p_j(t)}(t)$ or
$x_{p_i(t)}(t)= x_{p_j(t)}(t)$ and $p_i(t) < p_j(t)$. For every
$i$, we then let $y_i(t) = x_{p_i(t)}(t)$. The vector $y(t)$ is thus
sorted, so that 
$y_i(t)\leq y_j(t)$ for $i <j$. Let also $b_{ij}(t) =
a_{p_i(t)p_j(t)}(t)$. 
(This coefficient captures an interaction between the $i$th smallest and the $j$th smallest component of $x(t)$.) Proposition \ref{prop:sort}, proved in the Appendix, states that $y(t)$ satisfies an equation of the same form as \eqref{eq:evol_x}:
\begin{equation}\label{eq:evol_y}
y_i(t) = y_i(0) + \int_{0}^t \sum_{j=1}^n b_{ij}(\tau) \prt{y_j(\tau)
-y_i(\tau)}  d\tau, %\qquad 
\end{equation}
\rrev{for $i=1,\ldots,n.$} The definition of the functions $b_{ij}$ implies that
they also satisfy the cut-balance condition (\ref{eq:cut_balance_intro}). We
now define $S_m(t)$ to be a weighted sum of the values of the
first  $m$ (sorted) agents:
\begin{align*}
S_{m}(t) &= \sum_{i=1}^m K^{-i}y_{i}(t)\\&= S_m(0) + \int_0^t
{\sum_{i=1}^mK^{-i} \sum_{j=1}^n b_{ij}(\tau) \prt{y_j(\tau)-y_i(\tau)}
}\, d\tau.
\end{align*}
It follows from Lemma \ref{lem:deriv_weighted_avg_cutbal},
that the integrand is always nonnegative, so that
$S_m(t)$ is nondecreasing. 
Moreover, since all $b_{ij}(t)$ are nonnegative, Eq.\ 
(\ref{eq:evol_y}) can be used to show that $y_1(0) \leq y_i(t) \leq y_n(0)$, 
for all $i$ and $t$. In particular, each $S_m(t)$ is bounded above 
and therefore converges. This implies, using induction on $i$, that every $y_i(t)$
converges to a limit $y_i^*=\lim_{t\to \infty}y_i(t) \in
[y_1(0),y_n(0)] = [\min_jx_j(0),\max_jx_j(0)]$. Using the continuity of
$x$ and the definition of $y$, it is then an easy exercise to show that each $x_i(t)$ must also converge to one of the values 
$y_j^*$. This concludes the proof of part (a) of the theorem.

%%%%%%%%%%%%%%%%%%%%%%%%%%%%%%%%%%%%%%%%%%%%%%%%%%%%%%%%%%%%%%%%%%%%%%%%%%%%%%%%%%%%%%%%
We \TN{now prove part \rrev{(b)}.}
For every $m=1,\dots,n$, since \rrev{$S_m(t)$} converges to some $S_m^*$, we have  
\begin{equation}\label{eq:variation_S_m_av}
 \int_0^\infty {\sum_{i=1}^mK^{-i} \sum_{j=1}^n b_{ij}(t) \prt{y_j(t)-y_i(t)}\,
} dt  = S_m^*-S_m(0) <\infty.
\end{equation}
The integrand in this expression can be  rewritten as
\begin{align} 
\sum_{i=1}^mK^{-i} \Bigg(&\sum_{j=1}^m b_{ij}(t) \prt{y_j(t)-y_i(t)} \notag\\
 &+\sum_{j=m+1}^n b_{ij}(t) \prt{y_m(t)-y_i(t)} \Bigg) \notag\\
 + \sum_{i=1}^mK^{-i} &\sum_{j=m+1}^n  b_{ij}(t) \prt{y_j(t)-y_m(t)} 
\end{align} 
It follows from Lemma \ref{lem:deriv_weighted_avg_cutbal} applied to the coefficients $b_{ij}$ and the  sorted vector $(y_1(t),y_2(t),\dots,y_m(t),$ $ y_m(t),\dots,y_m(t))$ that the first term in the sum above is nonnegative, and thus that 
\begin{align*}
&\sum_{i=1}^mK^{-i} \sum_{j=m+1}^n b_{ij}(t) \prt{y_j(t)-y_m(t)}\\
\leq
&\sum_{i=1}^mK^{-i} \sum_{j=1}^n b_{ij}(t) \prt{y_j(t)-y_i(t)}.
\end{align*}
Equation (\ref{eq:variation_S_m_av}) then implies that
$$
\int_0^\infty\sum_{i=1}^m \sum_{j=m+1}^n K^{-i}b_{ij}(t) \prt{y_j(t)-y_m(t)} dt < \infty.
$$
Since  $K^{-i}b_{ij}(t)\geq 0$ and $y_j(t) \geq y_m(t)$ for $j>m$, every term of the sum in the integrand above is nonnegative, for every $t$. Then,  the boundedness of the integral implies the boundedness of every $\int_0^\infty b_{ij}(t) \prt{y_j(t)-y_m(t)} dt= \int_0^\infty b_{ij}(t) \abs{y_j(t)-y_m(t)} dt $ when $m<j$. A symmetrical argument shows that $\int_0^\infty b_{ij}(t) \abs{y_j(t)-y_m(t)} dt$ is also bounded when $j<m$, so that
\begin{equation}\label{eq:sum_y_bounded}
\int_0^\infty\sum_{i=1}^n\sum_{j=1}^nb_{ij}(t) \abs{y_j(t)-y_i(t)} dt<\infty.
\end{equation}
Because of the definitions $y_i(\tau) = x_{p_i(\tau)}(\tau)$ and
$b_{ij}(\tau) = a_{p_i(\tau)p_j(\tau)}(\tau)$, and the fact that $p(t)$ is a permutation, the equality
$$
\sum_{i=1}^n\sum_{j=1}^na_{ij}(t) \abs{x_j(t)-x_i(t)} = \sum_{i=1}^n\sum_{j=1}^nb_{ij}(t) \abs{y_j(t)-y_i(t)}
$$
holds for all $t$, which together with the nonnegativity of all $a_{ij}(t) \abs{x_j(t)-x_i(t)}$ and (\ref{eq:sum_y_bounded}) implies that $\int_0^\infty a_{ij}(t)\abs{x_j(t)-x_i(t)}dt <\infty$ for all $i,j$. 

Suppose now that the edge $(j,i)$ is in the graph $G$, i.e., that $\int_0^\infty a_{ij}(t)\,dt = \infty$. From part (a), we know that $\abs{x_i(t)-x_j(t)}$ converges to a constant value for every $i,j$.  Assumption   \ref {as:bi}
 ($\int_0^{t'}a_{ij}(t)\, dt<\infty$ for all $t'$) and the fact that $\int_0^\infty a_{ij}(t)\abs{x_j(t)-x_i(t)}dt <\infty$ imply that the value to which $\abs{x_i(t)-x_j(t)}$ converges must be 0, and thus that $x_i^ * = x_j^*$. 
If $i$ and $j$ are not directly connected, i.e., $(j,i)$ is not and edge in $G$,  but belong to the same connected component of $G$, the equality $x_i^ * = x_j^*$ follows by using transitivity along a path from $i$ to $j$.

%%%%%%%%%%%%%%%%%%%%%%%%%%%%%%%%%%%%%%%%%%%%%%%%%%%%%%%%%%%%%%%%%%%%%%%%%%%%%%%%%%%%%%%%
It remains to prove part (c). Consider a partition of the agents
in two groups, $V_1$ and $V_2$, that are disconnected in $G$, i.e., $\lim_{t\to
\infty}\int_{0}^t a_{ij}(\tau)\,d\tau < \infty$ for all $i\in V_1$ and $j\in V_2$,
and also for all $i\in V_2$ and $j\in V_1$. Thus, 
there exists some $t_{1/4}$ such that
$
\int_{t_{1/4}}^\infty\sum_{i\in V_1,j\in V_2} \prt{a_{ij}(\tau) +
a_{ji}(\tau)}d\tau < \frac{1}{4}.
$
We will first show that there exists a full-dimensional set of initial vectors $x(0)$ for which $\lim_{t\to\infty} x_i(t)\neq \lim_{t\to\infty} x_j(t)$,
when $i\in V_1$ and $j\in V_2$.

Since $\int_0^t a_{ij}(\tau)\,d\tau < \infty$, it can be proved that the system in Eq.\ \eqref{eq:evol_x} admits a unique solution and that the state transition (or fundamental) matrix, which maps the initial conditions $x(0)$ to $x(t)$, has full rank for any finite $t$; see \cite{Sontag:1998}, for example (specifically, Theorem 54, Proposition C3.8, appendix C3, and appendix C4). In particular, we can chose $x(0)$ such that $x_i(t_{1/4}) = 0$ if $i\in V_1$, and
$x_i(t_{1/4}) = 1$ if $i\in V_2$. Let $m$ be the number of agents
in $V_1$, and let $y$ be the sorted version of $x$ as above. There holds
$y_1(t_{1/4}) = \dots = y_m(t_{1/4}) = 0$ and $y_{m+1}(t_{1/4}) =
\dots = y_n(t_{1/4}) = 1$. 

Consider now a $t^*\geq t_{1/4}$ such that $y_m(t) < y_{m+1}(t)$ holds for all $t\in [t_{1/4},t^*]$. The continuity of $x$ and the definition of $y$ implies that for all $t\in [t_{1/4},t^*]$, we have $x_i(t) \leq y_m(t)$ for every $i\in V_1$, and $x_j(t) \geq y_{m+1}(t)$ for every $j \in V_2$.  In the same time interval, the permutation $p(t)$ that maps the indices of $x(t)$ to the corresponding indices of $y(t)$ takes thus values smaller than or equal to $m$ for indices $i\in V_1$ and larger than $m$ for indices $j\in V_2$. As a result, 
$\sum_{i\leq m, j>m}\prt{b_{ij}(t) +b_{ji}(t)}  = \sum_{i\in V_1, j\in V_2}\prt{a_{ij}(t) +a_{ji}(t)} $ for all $t\in  [t_{1/4},t^*]$. The definition of $t_{1/4}$ and the nonnegativity of the $a_{ij}(t)$ imply that, for all $t$ in that interval,
\begin{align}\label{eq:bounds_on_int_b}
\int_{t_{1/4}}^{t} \sum_{i=1}^{m}&\sum_{j=m+1}^n  \prt{b_{ij}(\tau) +b_{ji}(\tau)}d\tau \notag \\&= \int_{t_{1/4}}^{t}  \sum_{i\in V_1, j\in V_2}\prt{a_{ij}(\tau) +a_{ji}(\tau)}d\tau 
< \frac{1}{4}.
\end{align}

We now fix an arbitrary $t\in [t_{1/4},t^*]$, and show that $y_m(t)$ and $y_{m+1}(t)$ remain separated by at least $1/2$. 
Using Eq.\ (\ref{eq:evol_y}) and $y_m(t_{1/4})=0$,  we see that
\begin{align}\label{eq:above}
y_m(t)  & = \int_{t_{1/4}}^t \sum_{j=1}^mb_{mj}(\tau)
\prt{y_j(\tau) - y_m(\tau)}d\tau \notag \\ &+ \int_{t_{1/4}}^t \sum_{j=
m+1}^n b_{mj}(\tau) \prt{y_j(\tau) - y_m(\tau)}  d\tau.
\end{align}
The first term is non-positive by the definition of $y$. Consider now
the second term. Since $y_i(t_{1/4})\in [0,1]$ for all $i$, we {obtain}  
$y_i(\tau)\in [0,1]$ for all $t\geq t_{1/4}$, so that $y_j(\tau) -
y_m(\tau) \leq 1$ for every $j$.  Equations \eqref{eq:above} and \eqref{eq:bounds_on_int_b} then imply that 
$$
y_m(t) \leq   \int_{t_{1/4}}^t\sum_{j= m+1}^nb_{mj}(\tau)
\, d\tau< \frac{1}{4},
$$
for every $t \in  [t_{1/4},t^*]$. A similar argument shows that
$y_{m+1}(t)> 3/4$ for all $t\in  [t_{1/4},t^*]$. 
Recalling the definition of $t^*$, we have essentially proved that, after time $t_{1/4}$ and as long as $y_{m}(t)<y_{m+1}(t)$, we must have $y_{m+1}(t)-y_m(t)>1/2$. Because $y(t)$ is continuous, it follows easily that the inequality $y_{m+1}(t)-y_m(t)>1/2$ must hold for all times. Since we have seen that, for $t \in  [t_{1/4},t^*]$ , there holds  $x_j(t)> 3/4$ for all $j\in V_2$ and $x_i(t)< 1/4$  for all $i\in V_1$, this implies that  $x_i^*= \lim_{t\to\infty }x_i(t) \leq 1/4$ for $i\in V_1$ and $x_j^*=\lim_{t\to\infty }x_j(t) \geq 3/4$ for $j\in V_2$.

Note that the function that maps the initial condition $x(0)$ to $x^* =\lim_{t\to\infty}x(t)$ is linear; let $L$ be the matrix that represents this linear mapping.
We use $e_i$ to denote the {$i$th unit} vector in $\Re^n$.
We have shown above that if $i$ and $j$ belong to different connected components
of $G$, there exists at least one $x(0)$ for which
$x_i^*-x_j^* = (e_i-e_j)^TLx(0) \neq 0.$ Therefore,
$(e_i-e_j)^TL\neq 0$. In particular, the set of initial conditions $x(0)$ for
which $(e_i-e_j)^TLx(0) = 0$ is contained in an $n-1$ dimensional subspace of
$\Re^n$, which establishes part \rrev{(c)} of the theorem.
\end{proof}

We note that Theorem \ref{thm:main_convergence} 
has an analog for the case where each agent's value $x_i(t)$ is actually a multi-dimensional vector, obtained by applying Theorem 
\ref{thm:main_convergence} separately to each component.

The key assumption in Theorem \ref{thm:main_convergence},
which allows us to prove the convergence of the
$x_i$, is that the aggregate influence of a group of agents on the others
remains within a constant factor of the reverse aggregate influence.
As we will see in Section \ref{sec:part_cases} (Proposition \ref{cor:stronger_condition}),
the cut-balance assumption is a generalization (weaker version) of a more local type-symmetry condition. The latter condition requires that
$a_{ij}(t)$ be positive if and only if $a_{ji}(t)$ is positive,
and that the ratio of these two quantities be bounded by $K$. (A system satisfying the type-symmetry condition with $K$ automatically satisfies the cut-balance condition with the same $K$.) 
The requirement that the ratio be bounded is essential. 
To demonstrate this, we present an example where $a_{ij}(t) > 0$ whenever $a_{ji}(t) >0$, but for which convergence fails to hold.

\begin{exmp}\label{ex:1}
\rm
Let $n=3$ and consider the trajectories
\begin{equation}\label{eq:three_agents}
x_1(t) = 3 + e^{-t},\e x_2(t) = \sin(t), \e x_3(t) = - 3 -
e^{-t}.\end{equation}
These trajectories are a solution to  
(\ref{eq:system_diff}), for the case where
$$
a_{12}(t) = \frac{1}{e^t(3-\sin(t))+1}, \e a_{32}(t) =
\frac{1}{e^t(3+\sin(t))+1},
$$
$$
a_{21}(t) = \frac{1}{2} + \frac{\sin(t)+\cos(t)}{6+2e^{-t}},\e
a_{23}(t) = \frac{1}{2} - \frac{\sin(t)+\cos(t)}{6+2e^{-t}},
$$
and $a_{13}(t) = a_{31}(t)= 0$. 
Note that the trajectory ($x_2(t)$, in particular) does not converge. 
This system satisfies a weak
form of type-symmetry:  an agent cannot
influence another without being itself influenced in return, but
the ratio between these two influences can grow unbounded.
\hfill $\square$
\end{exmp}

One might speculate that the failure to converge in Example \ref{ex:1}
is due to the exponential growth of some of the 
ratios $a_{ij}(t)/a_{ji}(t)$, and that convergence might
still be guaranteed if these ratios were bounded by a slowly growing  function of time. However, 
this is not the case either: an example with fast-growing ratios is equivalent to one with slow-growing ratios, once we rescale the time axis.

More concretely, let $g(t)$ be a nonnegative increasing function that grows slowly to infinity, and let $z(t)=x(g(t))$, where $x(\cdot)$ is as in the preceding example. Then, $z(t)$ satisfies 
\eqref{eq:system_diff} with new coefficients
$\overline{a}_{ij}(t) = \dot g(t) a_{ij}(g(t))$.
The ratio  $\overline{a}_{ij}(t)/\overline{a}_{ji}(t)$ is a slowed down version
of the ratio $a_{ij}(t)/a_{ji}(t)$, and is therefore slowly growing.
On the other hand, since $x(t)$ does not converge,  $z(t)$ does not converge either.

%% file: part_cases.tex
%%%%%%%%%%%%%%%%%%%%%%%%%%%%%%%%%%%%%%%%%%%%%%%%%%%%%%%%%%
%%%%%%%%%%%%%%%%%%%%%%%%%%%%%%%%%%%%%%%%%%%%%%%%%%%%%%%%%%
%%%%%%%%%%%%%%%%%%%%%%%%%%%%%%%%%%%%%%%%%%%%%%%%%%%%%%%%%%
\section{Particular cases of cut-balanced dynamics}
%%%%%%%%%%%%%%%%%%%%%%%%%%%%%%%%%%%%%%%%%%%%%%%%%%%%%%%%%%
%%%%%%%%%%%%%%%%%%%%%%%%%%%%%%%%%%%%%%%%%%%%%%%%%%%%%%%%%%

\label{sec:part_cases}

The cut-balance condition is a rather weak assumption,
but may be hard to check. The next proposition provides five special cases of
cut-balanced systems that often arise naturally. It should however be understood that the class of cut-balanced systems is not restricted to these five particular
cases.

\begin{prop}\label{cor:stronger_condition} 
A collection of nonnegative coefficients $a_{ij}(\cdot)$ that satisfies
any of the following five conditions also satisfies the cut-balance condition (Assumption \ref{as:cb}).
\begin{enumerate}
\item[(a)] Symmetry: $a_{ij}(t) = a_{ji}(t)$, for all $i,j,t$.
\item[(b)] Type-symmetry: There exists $K\geq 1$ such that\\ $K^{-1}a_{ji}(t) \leq a_{ij}(t) \leq K a_{ji}\green{(t)}$, for all $i,j,t$.
\item[(c)] Average-preserving dynamics: $\sum_{j}a_{ij}(t) =
\sum_{j}a_{ji}(t)$, for all $i,t$.
\item[(d)] Weighted average-preserving dynamics: There exist $w_i>0$ such
that $\sum_j w_i a_{ij}(t)=\sum_j w_ja_{ji}(t) $, for all \green{$i,t$.}
\item[(e)] Bounded coefficients and  \green{set-symmetry}: %\comjh{find a name} 
There \green{exist $M$ and $\alpha$ with} $M\geq \alpha >0$ such that for all $i,j,t$ either $a_{ij}(t)=0$ or $a_{ij}(t) \in [\alpha,M]$; and, for any %nonempty proper 
subset $S$ of $\{1,\ldots,n\}$,  there exist
$i\in S$ \green{and} $j\not \in S$ with $a_{ij}(t)>0$ if and only if there
exist $i'\in S$ \green{and} $j'\not \in S$ with $a_{j'i'}(t)>0$.

\end{enumerate}
\end{prop}

\begin{proof}
Condition (a) implies condition (b), with $K=1$. If condition (b) holds, then by summing over all $i$ in some set of nodes $S$ and all  $j\green{\notin S}$, we obtain the cut-balance condition.

Condition (c) implies condition (d), with
$w_i=1$. \TN{Suppose that condition (d) holds.} \green{Then,}  
\begin{align*}\sum_{i\in S,j\notin S}w_ja_{ji}(t) &=  \sum_{i\in
S}\sum_{j=1}^n w_ja_{ji}(t) - \sum_{i\in S,j\in S}w_ja_{ji}(t)\\& =
\sum_{i\in S}\sum_{j=1}^nw_ia_{ij}(t) - \sum_{i\in S,j\in
S}w_ia_{ij}(t). 
\end{align*}
It follows that $\sum_{i\in S,j{\notin S}}w_ja_{ji}(t) = \sum_{i\in
S,j{\notin S}}w_ia_{ij}(t)$, and thus that
$$
\sum_{i\in S,j\notin S}a_{ji}(t) \geq \frac{\min_iw_i}{\max_iw_i}
\sum_{i\in S,j\notin S}a_{ij}(t).
$$
A reverse inequality follows from a symmetric{al} argument.
Therefore, 
the cut-balance condition holds with $K =
{\max_iw_i}/{\min_iw_i}$, which is well defined and no less than 1 because $w_i>0$ for
all $i$.

Finally, suppose that the condition (e) is satisfied, and consider a set $S$ and a time $t$. If $a_{ij}(t)=0$ for all $i\in S$ and $j\not\in S$, then (e) implies that $a_{ji}(t)=0$ for all $i\in S$ and $j\not\in S$ so that $\sum_{i\in S,j\notin S}a_{ij}(t)= 0 =\sum_{i\in S,j\notin S}a_{ji}(t)$, and the cut-balance condition is trivially satisfied for that set $S$ and any $K$. If on the other hand there exists $i\in S,j\not\in S$ for which $a_{ij}(t)>0$, then (e) implies the existence of $i'\in S, j'\not\in S$ such that $a_{j'i'}(t)>0$, and $a_{j'i'},a_{ij}\in [\alpha,M]$. Let $\abs{S}$ be the cardinality of $S$. Then, 
$$
\abs{S}(n-\abs{S}) M \geq \sum_{i\in S,j\notin S}a_{ij}(t)\geq \alpha,$$
and% ,  \text{\qquad  and \qquad} 
$$\abs{S}(n-\abs{S}) M \geq \sum_{i\in S,j\notin S}a_{ji}(t)\geq \alpha, 
$$
so that the cut balance condition holds with $K= \max_S \frac{\abs{S}(n-\abs{S}) M}{\alpha}\leq n^2\frac{M}{4\alpha}$.
\end{proof}
Note that condition (d) remains sufficient for cut-balance if the weights $w_i$ change with time, provided that the ratio $(\max_i w_i(t)) /(\min_i w_i(t))$ remains uniformly bounded (the {same proof applies}). We also note that the connectivity condition in (e) is equivalent to requiring every weakly connected component to be strongly connected in the graph obtained by connecting $(j,i)$ if $a_{ij}(t)>0$, for every $t$.

%% file: random.tex
%%%%%%%%%%%%%%%%%%%%%%%%%%%%%%%%%%%%%%%%%%%
%%%%%%%%%%%%%%%%%%%%%%%%%%%%%%%%%%%%%%%%%%%
%%%%%%%%%%%%%%%%%%%%%%%%%%%%%%%%%%%%%%%%%%%
\section{Application to systems with random interactions}
%%%%%%%%%%%%%%%%%%%%%%%%%%%%%%%%%%%%%%%%%%%
%%%%%%%%%%%%%%%%%%%%%%%%%%%%%%%%%%%%%%%%%%%
\label{sec:random}

We \green{give} a brief discussion of systems with random interactions. Consensus seeking systems where interactions are determined by a random process have been the object of \green{several recent} studies. For example,  Matei et al.\ \cite{MateiMartinsBaras:2008}   consider  the case where the matrix of coefficients $a_{ij}(t)$ follows a (finite\green{-state}) irreducible Markov process, and is always average-preserving. They prove that the system converges almost surely to consensus for all initial conditions if and only if the union of the graphs corresponding to each of the states of the Markov chain is strongly connected. This result is extended to continuous-time systems in \cite{MateiBaras:2009}. In \cite{TahbazJadbabaie:2010}, Tahbaz-Salehi and Jadbabaie consider discrete-time consensus-seeking systems where the interconnection is generated by an ergodic and stationary random process, without assuming that the average is preserved. They prove that the system converges almost surely to consensus  if and only if an \green{associated average graph} contains a directed spanning tree.

It turns out that convergence for the case of random interactions is a simple consequence of deterministic convergence results: Theorem \ref{thm:main_convergence} can be directly applied to systems where the coefficients $a_{ij}(\cdot)$ are modeled as  a random process whose sample path satisfies the cut-balance condition with probability 1 (possibly with a different constant $K$ for different sample paths, and even in the absence of global upper bound on $K$). Indeed, if this is the case, Theorem \ref{thm:main_convergence} implies that each $x_i(t)$ converges, with probability 1.
Furthermore, if  ${\rm P}(\int_0^\infty a_{ij}(t)dt =\infty) =1$, then $x_i^*=x_j^*$,  with probability~1.

%% file: endogenous.tex
%%%%%%%%%%%%%%%%%%%%%%%%%%%%%%%%%%%%%%%%%%%%%%%%%%%%%%%%%%%%%%
%%%%%%%%%%%%%%%%%%%%%%%%%%%%%%%%%%%%%%%%%%%%%%%%%%%%%%%%%%%%%%
%%%%%%%%%%%%%%%%%%%%%%%%%%%%%%%%%%%%%%%%%%%%%%%%%%%%%%%%%%%%%%
\section{Application to systems with endogenous connectivity}\label{sec:application_endogenous}
%%%%%%%%%%%%%%%%%%%%%%%%%%%%%%%%%%%%%%%%%%%%%%%%%%%%%%%%%%%%%%
%%%%%%%%%%%%%%%%%%%%%%%%%%%%%%%%%%%%%%%%%%%%%%%%%%%%%%%%%%%%%%

Theorem \ref{thm:main_convergence} dealt with the case where the coefficients
$a_{ij}(t)$ are given functions of time; in particular, $x(t)$ was generated by a \emph{linear}, albeit time-varying, differential or integral equation. We now show that Theorem \ref{thm:main_convergence} also applies  to \emph{nonlinear} systems where the coefficients (and the interaction
topology) are endogenously determined by the vector $x(t)$ of agent values.
This is possible because Theorem \ref{thm:main_convergence} allows for arbitrary variations of the coefficients $a_{ij}(t)$, thus encompassing the endogenous case.

\begin{cor}\label{cor:endogenous}
For every $i$ and $j$, we are given a nonnegative measurable function $a_{ij}:\Re^+\times \Re^n\to \Re^+$. 
Let $x:\Re^+\to \Re^n$ be a \green{measurable function that satisfies} 
the system of integral equations
\begin{equation}\label{eq:evol_x_endogenous}
x_i(t) = x_i(0) + \int_{0}^t \sum_j
a_{ij}(\tau,x(\tau))\prt{x_j(\tau) - x_i(\tau)}d\tau, 
\end{equation}
for $ i=1,\ldots,n.$ Suppose that there exists $K\geq 1$ such that for all $x$, and $t$,
and any nonempty proper subset $S$
of $\{1,\dots,n\}$, we have 
\begin{equation*}
K^{-1}\sum_{i\in S,j\notin S}  a_{ji}(t,x)\leq \sum_{i\in
S,j\notin S} a_{ij}(t,x) \leq  K\sum_{i\in S,j\notin S} a_{ji}(t,x).
\end{equation*}

\begin{enumerate}
\item[(a)]
The limit $x_i^* = \lim_{t\to \infty}x_i(t)$ exists, and $x_i^*\in
[\min_jx_j(0),\max_jx_j(0)]$.
\item[(b)]
If $i$ and $j$ belong to the same connected component of $G$, 
then $x_i^* = x_j^*$, where the unbounded interactions graph $G$ is defined for each trajectory by letting $(j,i)\in E$ if $\int_{0}^\infty
a_{ij}(t,x(t))\,dt = \infty$.
\end{enumerate}
\end{cor}

\begin{proof}
Let us fix a solution $x$ to Eq.\ \eqref{eq:evol_x_endogenous}.
For this particular function $x$, and for every $i$, $j$, we define a 
\green{(necessarily measurable)} 
function $\hat
a_{ij}:\Re^+\to \Re^+$ by letting $\hat a_{ij}(t) = a_{ij}(t,x(t))$.
By the assumptions of the corollary, the functions $\hat a_{ij}$ satisfy the cut-balance
condition (Assumption \ref{as:cb}). Furthermore, $x$ is also a solution to
the system of (linear) integral equations
$$
x_i(t) = x_i(0) +\int_{0}^t\hat
a_{ij}(\tau)\prt{x_j(\tau)-x_i(\tau)}d\tau,
$$
for $i=1,\ldots,n.$
The result follows by applying  Theorem
\ref{thm:main_convergence} to the latter system.
\end{proof}

Note that the nonlinear system of integral equations \eqref{eq:evol_x_endogenous} considered in Corollary \ref{cor:endogenous} may have zero, one, or multiple solutions. Our result does not have any implication on the problem of existence or uniqueness of a solution, but applies  to every solution, if one exists.
Naturally, Corollary \ref{cor:endogenous} also holds if the coefficients $a_{ij}(x,t)$
satisfy stronger conditions such as type-symmetry or the condition $\sum_j
w_ja_{ji}(t,x)= \sum_j w_i a_{ij}(t,x)$ for some positive coefficients  $w_i$, as in Proposition
\ref{cor:stronger_condition}. 

We note that part (c) of Theorem
\ref{thm:main_convergence} does not extend to the nonlinear case
where the coefficients $a_{ij}$ also depend on $x$. Indeed, the
proof of Corollary \ref{cor:endogenous} applies Theorem
\ref{thm:main_convergence} to an auxiliary linear system,
and the choice of this linear system is based on the actual solution $x(\cdot)$.
Part (c) of Theorem
\ref{thm:main_convergence} does apply to this particular linear system, and implies that 
$x_i^*$ is indeed different from $x_j^*$ whenever $i$ and $j$ belong to
different connected components of the associated graph $G$, unless
$x(0)$ belongs to a lower-dimensional exceptional set.
However, this exceptional set is associated with the particular linear system, which is in turn determined by $x(0)$; different $x(0)$ can be associated with different exceptional sets $D(x(0))$. So, it is in principle possible that 
every $x(0)$ in a full-dimensional set falls in the exceptional set $D(x(0))$. This is not just a theoretical possibility, as
illustrated by the four-dimensional example that follows.

\begin{exmp}\label{ex:2}
\rm
Let $n=4$. Consider a sorted initial vector, so that $x_1(0)\leq x_2(0)\leq x_3(0)\leq x_4(0)$. Suppose that the coefficients $a_{ij}$ have no explicit dependence on time, but are functions of $x$, with 
$a_{13}(x) = a_{31}(x) =1$ and 
$a_{24}(x) = a_{42}(x) = 1$ as long $x_1<x_2<x_3<x_4$.
Otherwise, $a_{13}(x) = a_{31}(x)=a_{24}(x) = a_{42}(x) = 0$. All
other coefficients are 0. 
These coefficients are
symmetric, and thus cut-balanced. 
The corresponding system has a solution of the following form: 
$x_1(t),x_2(t)$ keep increasing and $x_3(t),x_4(t)$ keep decreasing, 
until some time $t^*$ at which agents 2 and 3 hold the same value; after that
time, all values remain constant. Thus, there is a 
$4$-dimensional set of initial conditions for which the resulting limits satisfy $x_2^*=x_3^*$.
Note that
$$\int_{0}^\infty a_{ij}(t)\, dt =  \int_{0}^{t^*} a_{ij}(t)\, dt <\infty,\qquad  \mbox{for all }i,j,
$$
and
the unbounded interactions graph $G$ has no edges.
Yet, despite the fact that nodes 2 and 3 belong to different strongly connected components, $x_2^*$ and $x_3^*$ are equal on a $4$-dimensional set of initial conditions. $\square$
\end{exmp}

Finally, we note that the structure of the graph $G$ in Corollary \ref{cor:endogenous} can be hard to determine, because $G$ depends on the evolution of $x$ via the $a_{ij}(t,x(t))$, and the evolution of $x$ is a priori unknown. In particular, it may be hard to determine  whether $G$ will be connected, thus guaranteeing consensus. However, as will be illustrated in the application below, the first part of the Corollary guarantees the convergence of any system satisfying the cut-balance condition. One can then use additional information on the graph $G$ to characterize the possible limiting states $x^*$.

We now apply Corollary
\ref{cor:endogenous} to 
a nonlinear multi-agent system of a form studied in 
\cite{CanutoFagnaniTilli:2008,Lorenz:2007,
BlondelHendrickxTsitsiklis:2009_APCT,HegselmannKrause:2002,Ben-naimKrapivskyRedner:2003,DeffuantNeauAmblardWeisbuch:2000,MirtabatabaeiBullo:2011,LinBrouckeFrancis:2004} (often in the context of bounded confidence models) in which
the agent values evolve according to the integral equation
version of
\begin{equation}\label{eq:APCT}
\green{\frac{d}{dt}}x_i(t) = \sum_{j:\, |x_i(t)-x_j(t)|<1}\prt{x_j(t)-x_i(t)}.
\end{equation}
The evolution of the interaction topology for this system is a
priori unknown, because it depends on the a priori unknown evolution of $x$. In addition, the interaction topology can, in principle, 
change an infinite number of times during a finite time interval. Determining whether such a system converges can be complicated. And indeed, the convergence of an asymmetric counterpart of (\ref{eq:APCT}) remains open \cite{MirtabatabaeiBullo:2011}.
Observe now that (\ref{eq:APCT}) is of the form 
(\ref{eq:evol_x_endogenous}), with $a_{ij}(x) = 1$ if
$\abs{x_i-x_j}<1$, and  $a_{ij}(x) =0$ otherwise. 
The coefficients $a_{ij}$ are symmetric and therefore satisfy 
the cut-balance condition in Corollary \ref{cor:endogenous}. 
Part (a) of the corollary implies that
the limit $x_i^*=\lim_{t\to \infty} x_i(t)$
exists for every $i$. 
Suppose now that for
some $i,j$, we have $\abs{x_i^*-x_j^*}<1$. Then, there exists a time after which $|x_i(t) -x_j(t)|
<1$ and therefore $a_{ij}(x(t))=1$. As
a consequence, $\int_{0}^\infty a_{ij}(x(t))\, dt
=\infty$, and Corollary \ref{cor:endogenous}(c) implies that
$x_i^*=x_j^*$. This proves that the system converges, and
that the limiting values of any two agents are either equal or
separated by at least 1, a result which had been obtained by ad hoc arguments in \cite{Hendrickx:2008phdthesis}. 

Exactly the same argument can be made for
a system that evolves according to the integral equation version of
\begin{equation*}
\green{\frac{d}{dt}} x_i(t) =
\frac{\sum_{j:\,|x_i(t)-x_j(t)|<r}\prt{x_j(t)-x_i(t)}}{\sum_{j:\, |x_i(t)-x_j(t)|<r}1}.
\end{equation*}
(We let $\dot x_i(t)=0$ whenever the denominator on the right-hand side is zero.)
This system satisfies a type-symmetry condition with $K=N$.
A variant of \green{such a system, with a different interaction radius $r_i$ for each $i$,} has been studied in \cite{LinBrouckeFrancis:2004} under the assumption that the graph of interactions is strongly connected at every $t$.

\green{A further} variation of (\ref{eq:APCT}) is of the form 
\begin{equation}\label{eq:APCTwithF}
\green{\frac{d}{dt}}  x_i(t) = \sum_{j} f(x_j(t)-x_i(t))\prt{x_j(t)-x_i(t)},
\end{equation}
where $f$ is \green{an even nonnegative} function. A
multidimensional version of (\ref{eq:APCTwithF}), 
\green{where each $x_i$ is a vector,} is studied in
\cite{CanutoFagnaniTilli:2008},
for the special case of a radially decreasing function $f$ that becomes zero beyond a
certain threshold. (The results in \cite{CanutoFagnaniTilli:2008} also allow for a continuum of agents, which arise for example when 
studying discrete-agent models in the limit of a large number of agents).
The system (\ref{eq:APCTwithF}) is of the form
\eqref{eq:evol_x_endogenous}, 
with $a_{ij}(x) = f(x_i-x_j)$. It satisfies a
type-symmetry condition, with $K=1$, and Corollary
\ref{cor:endogenous} implies convergence. Moreover, if
$f$ is bounded and is continuous except on a finite set, then for
any $i,j$, either $x_i^*=x_j^*$, or $x_i^*-x_j^*$ belongs to the
closure of the set $\{z:f(z)=0\}$ of roots of $f$. To see this, 
Corollary \ref{cor:endogenous} asserts that if
$x_i^*\neq x_j^*$, then $\int_{0}^\infty f(x_i(t)-x_j(t))\, dt
<\infty$, which implies that $x_i(t)-x_j(t)$ cannot stay forever
in a set on which $f$ admits a
positive lower bound.

%% file: discrete-time.tex
\section{Discrete-time systems}\label{sec:discrete-time}

Much of the literature on consensus-seeking processes is focused on
discrete-time systems. Typical results guarantee convergence to consensus under
the assumption that the system is \quotes{sufficiently connected}
on any time interval of a certain length \cite{Tsitsiklis:84phdthesis,Moreau:2005,JadbabaieLinMorse:2003}  
and \green{sometimes} provide bounds
on the convergence rate. When interactions are type-symmetric,
convergence to consensus is guaranteed under the weaker assumption
that the system remains \quotes{sufficiently connected} after any finite
time \cite{BlondelHendrickxOlshevskyTsitsiklis:2005,Moreau:2005,LiWang:2004} and results 2.5.9 and 2.6.2 in \cite{Lorenz:2003diplomathesis}.
One can then easily deduce that type-symmetric systems
always converge to a limit, at which we have consensus within each of possibly many agent clusters  
\cite{Hendrickx:2008phdthesis,Lorenz:2005,HendrickxBlondel:2006_MTNS}.

In this section, we show that the convergence proof in 
\cite{BlondelHendrickxOlshevskyTsitsiklis:2005,Hendrickx:2008phdthesis}
can be extended easily to prove that cut-balance is also a sufficient condition for convergence in the discrete-time case,
as in Theorem \ref{thm:main_convergence}. A special 
case of this result asserts the convergence of systems
that preserve some weighted average of the states, and \green{thus includes  
a sample path} version of recent results of
\cite{TouriNedic:2010} on stochastic consensus-seeking systems.

Discrete-time systems are in some sense simpler because of
the absence of Zeno behaviors or unbounded sets of
finite measure. However, discrete-time systems allow for large instantaneous
variations of the agents' values. In particular, an agent could
entirely \quotes{forget} its value at time $t$ when computing its
value at time $t+1$, leading to instabilities where agents keep switching
their values. For this reason, we introduce two additional
assumptions: each agent is influenced by its own value
when computing its new value, and every positive coefficient must
be larger than some fixed positive lower bound.

\begin{thm}\label{thm:DTcutbalance}
Let $x:\mathbb{N}\to \Re^n$ satisfy
$$
x_i(t+1) = \sum_{j=1}^na_{ij}(t)x_j(t),\qquad i=1,\ldots,n,
$$
where $a_{ij}(t)\geq 0$ for all $i$, $j$, and $t$, and $\sum_{j=1}^n a_{ij}(t) =1$
for all $i$ and $t$. Assume that the following three  conditions hold.
\begin{itemize}
\item[(a)]
Lower bound on positive coefficients: there  exists some $\alpha >0$
such that if $a_{ij}(t)>0$, then
$a_{ij}(t)\geq\alpha$, for all $i$, $j$, and $t$.
\item[(b)]
Positive diagonal coefficients: we have $a_{ii}(t) \geq \alpha$, for all
$i$ and $t$.
\item[(c)] Cut-balance: for any nonempty proper subset $S$ of $\{1,\ldots,n\}$,  there exist
$i\in S$ and $j\not \in S$ with $a_{ij}(t)>0$ if and only if there
exist $i'\in S$ and $j'\not \in S$ with $a_{j'i'}(t)>0$.
\end{itemize}
Then, the limit $x_i^* = \lim_{t\to \infty}x_i(t)$ exists, and $x_i^*\in
[\min_jx_j(0),\max_jx_j(0)]$. Furthermore, consider the directed graph
$G=(\{1,\dots,n\},E)$ in which $(j,i)\in E$ if $a_{ij}(t)>0$ infinitely
often. Then, every weakly connected component of $G$ is strongly
connected, and if $i$ and $j$ belong to the same connected
component of $G$, then $x_i^* = x_j^*$.
\end{thm}
\begin{proof}
The fact that every weakly connected component of $G$ is strongly
connected is proved exactly as in Theorem
\ref{thm:main_convergence}. 

Consider such a connected component $C$. It
follows from the definition of $G$ that there exists a time
$t^*$ after which $a_{ij}(t)=a_{ji}(t)=0$ for any $i\in C$ and $j\not \in C$. Thus, 
the values $x_i(t)$ with $i\in C$
do not influence and are not influenced by the remaining values
after time $t^*$. In particular, if $t^*\leq t' < t$, then $\min_{j\in
C} x_j(t')\leq x_i(t) \leq \max_{j\in C}x_j(t')$ holds for all
$i\in C$; furthermore, $\max_{i\in C} x_i(t)$ and $\min_{i\in C}x_i(t)$
are monotonically nonincreasing and nondecreasing, respectively.

We now show that \green{there exists a constant $\gamma>0$ such that} for any $t'\geq t^*$, there exists a $t''>t$
\green{for which} $\max_{i\in C}x_i(t'') -\min_{i\in C} x_i(t'')\leq \gamma
\big(\max_{i\in C}x_i(t') -\min_{\green{i\in C}} x_i(t')\big)$. 
\green{We assume that $|C|\geq 2$, because otherwise the claim is trivially true.}

Suppose that  $\max_{i\in C}x_i(t')=1$ and $\min_{i\in C}x_i(t')
=0$. This is not a loss of generality; the argument can be
carried out for any other values by appropriate scaling and
translation. For any $t$, let $C_k(t)$ be the set of indices
$i\in C$ \green{for which} $x_i(t) \geq \alpha^k$. Clearly, $C_0(t')$ is
nonempty. Consider some $t$ and $k$ such that $\emptyset\neq
C_k(t) \neq C$. We distinguish two cases. 

(i) Suppose that $a_{ij}(t)= 0$ for all  $i\in C_k(t)$ and $j\in
C\setminus C_{k}(t)$. Then, for any $i\in C_k(t)$, we have 
\begin{align*}
x_i(t+1) &= \sum_{j=1}^na_{ij}(t) x_j(t) = \sum_{j\in
C_k(t)}a_{ij}(t) x_j(t)\\ & \geq \sum_{j\in C_k(t)}a_{ij}(t) \alpha^k
\geq \alpha^k.
\end{align*}
(We have used here the facts that
$a_{ij}(t)=0$ for every $j\not\in C_k(t)$, and
$\sum_{j}a_{ij}(t) =1$.) Therefore, $i$ belongs to
$C_k(t+1)$ as well. So, in this case we have
$C_k(t) \subseteq C_{k}(t+1)$.

(ii)
Suppose now that 
$a_{ij}(t) >0$ for some $i\in C_k(t)$ and $j\in C\setminus C_{k}(t)$. 
Then the cut-balance condition, together with $t\geq t^*$,
implies that $a_{i'j'}>0$ for at least one $i'\in C\setminus
C_k(t)$ and $j'\in C_k(t)$. For this $i'$, we have
\begin{align*}
x_{i'}(t+1) &= \sum_{j=1}^na_{i'j}(t) x_j(t)  = \sum_{j\in
C}a_{i'j}(t) x_j(t)\\ &\geq a_{i'j'}(t) x_{j'}(t) \geq \alpha \cdot \alpha^k
= \alpha^{k+1}
\end{align*}
where we have used the fact that $x_{\jnt{j}}(t)\geq  \min _{i\in
C}x_i(t')\geq 0$, for all $\jnt{j}\in C$ and $t\geq t'$. Therefore, $i'\in
C_{k+1}(t+1)$. Moreover, for any $i\in C_{k}(t)$, we have 
$x_i(t)=\sum_{j\in C}a_{ij}(t)x_j(t) \geq a_{ii}(t)x_i(t) \geq
\alpha \cdot \alpha^k = \alpha^{k+1}$, because $a_{ii}(t)\geq \alpha$
for all $i$ and $t$. Thus, if $a_{ij}(t)>0$ for some $i\in
C_{k}(t)$ and $j\in C\setminus C_{k}(t)$, then the set $C_{k+1}(t+1)$ contains $C_k(t)$ and
at least one additional node.

Recall now that $C_0(t')$ is
nonempty. Moreover, the definition of $C$ as a strongly connected
component of $G$ implies that for any $t$ and any nonempty set
$S\subset C$, there exists a $\hat t$ and some $i\in S$, $j\in S\setminus C$, such that $a_{ij}(\hat t)>0$. \green{Then, a straightforward inductive argument based on the above two cases} 
shows the existence of a time $t''>t'$ at which
$C_{\abs{C}-1}(t'') =C$, i.e., a time $t''$ at which $\min_{i\in
C}x_i(t'')\geq \alpha^{\abs{C}-1}$. Since $x_i(t)$ remains less 
than or equal to 1 for $i\in C$ and $t>t'$, we conclude that $\max_{i\in C}x_i(t'')- \min_{i\in C}x_i(t'')$ is bounded by
$$
(1-\alpha^{\abs{C}-1})\prt{\max_{i\in C}x_i(t')- \min_{i\in
C}x_i(t')}.
$$
This inequality, 
together with the fact that $\max_{i\in C}x_i(t)$ and
$\min_{i\in C}x_i(t)$ are respectively nonincreasing and
nondecreasing after time $t^*$, implies the convergence of
$x_i(t)$, for all $i\in C$, to a common limit.
\end{proof}

Observe that part (c) of Theorem \ref{thm:main_convergence},
convergence to generically different values for the different
components of $G$, has no counterpart for the discrete-time
case. Indeed, if $a_{ij}(1) = 1/n$ for all $i,j$, the system
reaches global consensus after one time step, irrespective of the connectivity properties of $G$. 

Condition (c) in Theorem \ref{thm:DTcutbalance} has a graph-theoretic  interpretation. For every $t$, let $G_t$ be the graph on $n$ nodes obtained by connecting $j$ to $i$ if $a_{ij}(t)$ is positive. Condition (c) is satisfied if and only if for every $t$, every weakly connected component of $G_t$ is strongly connected.

Finally, note that convergence results for discrete-time consensus seeking systems with random or endogenously determined interactions can be derived from  Theorem \ref{thm:DTcutbalance} in a straightforward manner, exactly as in Section \ref{sec:random} and Corollary \ref{cor:endogenous}, respectively.

%% file: conclusions.tex
%%%%%%%%%%%%%%%%%%%%%%%%%%%%%%%%%%%%%%%%%%%%%%%%%%%%%%%%%%%%%%
%%%%%%%%%%%%%%%%%%%%%%%%%%%%%%%%%%%%%%%%%%%%%%%%%%%%%%%%%%%%%%
%%%%%%%%%%%%%%%%%%%%%%%%%%%%%%%%%%%%%%%%%%%%%%%%%%%%%%%%%%%%%%
\section{Concluding remarks}\label{sec:ccl}
%%%%%%%%%%%%%%%%%%%%%%%%%%%%%%%%%%%%%%%%%%%%%%%%%%%%%%%%%%%%%%
%%%%%%%%%%%%%%%%%%%%%%%%%%%%%%%%%%%%%%%%%%%%%%%%%%%%%%%%%%%%%%

In this paper, we introduced a cut-balance condition, which is a natural and perhaps the broadest possible symmetry-like assumption 
for consensus seeking systems.
This assumption is satisfied, in particular, if the dynamics preserve a weighted average, or if no 
agent can influence another without incurring a proportional  reverse influence. 
We proved that the cut-balance assumption is a sufficient condition for the
convergence of continuous-time consensus seeking
systems, and provided a characterization of the resulting local consensus,
in terms of the evolution of the interaction
coefficients. We then applied our results to
systems with endogenously determined connectivity. Related results were also obtained for the discrete time case. 
We also showed that our result fails to hold if the proportionality constant $K$ in the cut-balance assumption can grow with time, without bound.

We end by discussing the possibility of extending our result to models involving a continuum of agents. Such models  appear naturally when studying discrete-agent models, in the limit of a large number of agents\cite{BlondelHendrickxTsitsiklis:2009_APCT,
CanutoFagnaniTilli:2008,Lorenz:2007,Ben-naimKrapivskyRedner:2003,FortunatoLatoraPluchinoRapisarda:2005}. 
\green{Let} $x_t(\alpha)$ be the value of agent $\alpha\in [0,1]$ at time
$t$. Consider then a function $x:[0,1]\times \Re^+\to
\Re:(\alpha,t)\to x_\alpha(t)$ that satisfies 
\begin{equation}\label{eq:continuum}
\frac{d}{dt}x_\alpha(t) = \int
a_{\alpha,\beta}\prt{x_\beta(t)-x_\alpha(t)}d\beta,
\end{equation}
for some measurable nonnegative function 
$a:[0,1]\times[0,1]\times\Re^+\to \Re^+:(\alpha,\beta,t)\to
a_{\alpha,\beta}(t)$. The 
extent to which our discrete-agent convergence result can be generalized to
the model (\ref{eq:continuum}) is an open problem, and the same is true in the discrete time-case. 

In addition to some technical difficulties inherent to treating a continuum of agents, the main reason \green{that} our approach does not apply directly to such systems is the following.  
Our proofs of Theorem \ref{thm:main_convergence}  and Lemma \ref{lem:deriv_weighted_avg_cutbal} rely \green{on a collection of functions $S_m$ of} the vector of agent states, \green{with} the property that $\frac{\partial S_m}{\partial x_i}\geq K\frac{\partial S_m}{\partial x_j} \geq 0$ if $j>i$, where $K\geq 1$ is the constant  from the cut-balance assumption. Finding nontrivial  functions having that property seems impossible in the case of a continuum of agents when $K>1$. Suppose indeed that we have  a function $S$ \green{of the configuration of agent values} such that 
\begin{equation}\label{eq:condition}
\frac{\partial S}{\partial x_\alpha}\geq K\frac{\partial S}{\partial x_\beta} \geq 0
\end{equation}
when $\beta > \alpha$. For  two agents $\alpha < \beta$, take now a sequence $\alpha < \gamma_1<\gamma_2< \dots< \gamma_n < \beta$. By a repeated application of (\ref{eq:condition})
we obtain $\frac{\partial S_m}{\partial x_\alpha}\geq K^{n+1}\frac{\partial S_m}{\partial x_\beta}$. Since this is true for arbitrary $n$, the derivative of $S$ with respect to the agent state would either be 0 everywhere or unbounded almost everywhere. \green{Whether an alternative approach can be used to establish unconditional convergence is an open problem.}

Let us also note that some 
convergence results for the continuum model are available for some special cases with $K=1$, for example because of symmetric interactions, i.e., $a_t(\alpha,\beta)=a_t(\beta,\alpha)$; see \cite{BlondelHendrickxTsitsiklis:2009_APCT,
CanutoFagnaniTilli:2008}).

%% file: appendix.tex
%%%%%%%%%%%%%%%%%%%%%%%%%%%%%%%%%%%%%%%%%%%%%%%%%%%%%%%%%%%%%%%%%%%%%%%%
%%%%%%%%%%%%%%%%%%%%%%%%%%%%%%%%%%%%%%%%%%%%%%%%%%%%%%%%%%%%%%%%%%%%%%%%
%%%%%%%%%%%%%%%%%%%%%%%%%%%%%%%%%%%%%%%%%%%%%%%%%%%%%%%%%%%%%%%%%%%%%%%%
\section{Evolution of the sorted vector}\label{appen:evol_sorted}
%%%%%%%%%%%%%%%%%%%%%%%%%%%%%%%%%%%%%%%%%%%%%%%%%%%%%%%%%%%%%%%%%%%%%%%%
%%%%%%%%%%%%%%%%%%%%%%%%%%%%%%%%%%%%%%%%%%%%%%%%%%%%%%%%%%%%%%%%%%%%%%%%

In this appendix we prove that if we sort the components of the vector $x$, 
the resulting vector satisfies essentially the same evolution equation as the original system, even though the required permutation can change with time as the relative order of different components changes. The result appears  elementary (if not obvious), yet we are not aware of a simple proof.

\begin{prop}\label{prop:sort}
Let $x:\Re^+ \to \Re^n$ be a continuous function
that satisfies 
\begin{equation}\label{eq:appen_evolx}
x_i(t) = x_i(0) + \int_{0}^t v_i(\tau)\, d\tau,
\qquad i=1,\ldots,n.
\end{equation}
For any $t\geq 0$, let 
$p(t)$ be the permutation of
$\{1,\dots,n\}$ defined by the following lexicographic rule: 
if $i<j$, then, either (i) $x_{p_i(t)}(t)< x_{p_j(t)}(t)$, or (ii) 
$x_{p_i(t)}(t)= x_{p_j(t)}(t)$ and $p_i(t) < p_j(t)$. Then, 
\begin{equation}\label{eq:evol_sorted}
x_{p_i(t)}(t) = x_{p_i(0)}(0) + \int_{0}^t
v_{p_i(\tau)}(\tau)\, d\tau.
\end{equation}
\end{prop}

Note that $x_{p_{i}(t)}(t)$ is a sorted version of
$x(t)$, with lexicographic tie breaking. Observe that when all $v_i(t)$ are smooth and the order of
$x$ changes only a finite number of times within a bounded interval, 
the result follows
immediately from the continuity of $x_{p_i(t)}(t)$ and the fact
that $\frac{d}{dt} x_{p_i(t)}(t) = v_{p_i(t)}(t)$ on every interval on
which $p_i(t)$ is constant. The proof that we present here is more
general, and only assumes measurability. In particular, it allows
for infinitely many discontinuities or order changes in  finite
time, something that cannot be ruled out, in general.

Our proof uses induction on $n$, starting with the particular
cases where $n=2$ or $n=3$.

\begin{lem}\label{lem:sort2}
Proposition \ref{prop:sort} holds when $n=2$.
\end{lem}

\begin{proof}
Let us fix a time $t>0$. We only give the proof for the case where
$x_1(0)\leq x_2(0)$, so that $p_1(0)=1$. (The proof for the case where 
$x_1(0)> x_2(0)$ is almost the same.)

We start by considering the case where
$p_1(t)= 1$.
Equation
(\ref{eq:appen_evolx}),  
applied to $i=1$, yields
$$x_{1}(t) = x_{1}(0)  + \int_{0}^tv_{1}(\tau)\, d\tau.$$
We
define $T_1=\{\tau\in[0,t]: x_1(\tau) \leq x_2(\tau)\}$ and
$T_2=\{\tau\in[0,t]: x_1(\tau) >x_2(\tau)\}$. Thus, $p_1(\tau) = 1$ for all $\tau\in T_1$ (including $\tau=0$), and $p_1(\tau) = 2$ for all $\tau\in T_2$. Moreover, since
$x$ is continuous, $T_1$ and $T_2$ are measurable sets. 
Therefore,
\begin{align*}
x_{1}(t) &= x_{1}(0)  + \int_{\tau \in
T_1}v_{1}(\tau)\, d\tau + \int_{\tau \in T_2}v_{1}(\tau)\,
d\tau \\&= x_{p_1(0)}(0)  + \int_{\tau \in T_1}v_{p_1(\tau)}(\tau)\, d\tau
+ \int_{\tau \in T_2}v_{1}(\tau)\, d\tau 
\end{align*}
The continuity of $x$, and a fortiori of $x_1-x_2$, implies
that $T_2$ is the union of an at most countable collection of disjoint open
intervals 
$(a_k,b_k)$, with $0<a_k<b_k<t$, $x_2(a_k) =
x_1(a_k)$, and $x_2(b_k) = x_1(b_k)$. For every such interval,
we have 
\begin{align*}
\int_{a_k}^{b_k}v_1(\tau)\, d\tau &= x_1(b_k) - x_1(a_k)\\ & = x_2(b_k) -
x_2(a_k)\\& = \int_{a_k}^{b_k}v_2(\tau)\, d\tau \\& = \int_{a_k}^{b_k}v_{p_1(\tau)}(\tau)\, d\tau,
\end{align*}
which implies that
\begin{align*}
\int_{\tau\in T_2}v_1(\tau)\,d\tau &= \sum_k
\prt{\int_{a_k}^{b_k}v_1(\tau)\,d\tau} \\& = \sum_k
\prt{\int_{a_k}^{b_k}v_{p_1(\tau)}(\tau)\, d\tau}\\ &= \int_{\tau \in
T_2}v_{p_1(\tau)}(\tau)\, d\tau.
\end{align*}
It follows that
\begin{align*}
x_{p_1(t)}(t)& =x_1(t)\\& =x_{1}(0)
 + \int_{\tau \in T_1}v_{p_1(\tau)}(\tau)\, d\tau
+ \int_{\tau \in T_2}v_{p_1(\tau)}(\tau)\, d\tau 
\\ &=x_{p_1(0)}+\int_0^t v_{p_1(\tau)}\,d\tau,
\end{align*}
as claimed.

Suppose now that $p_1(t) = 2$. Let $t^{*} = \max
\{\tau \leq t: x_1(\tau)\leq x_2(\tau) \}$; the maximum is attained because 
$x_1(0)\leq x_2(0)$ and $x$ is continuous.
Furthermore, $x_1(t^*)=x_2(t^*)$, and $p_1(\tau)=2$ for $\tau\in(t^*,t]$. 
We have $p_1(t^{*})=1$, so applying
the result we have proved above, with the interval $[0,t^*]$ replacing $[0,t]$, in the fourth equality below, we obtain
\begin{eqnarray*}
x_{p_1(t)}(t) &=& x_2(t)\\
&=&x_2(t^*) +\int_{t^*}^t v_2(\tau)\, d\tau\\
&=&x_{p_1(t^*)}(t^*) +\int_{t^*} ^t v_2(\tau)\, d\tau\\
&=&x_{p_1(0)}(0) +\int_0^{t^*} v_{p_1(\tau)}(\tau)\, d\tau+\int_{t^*}^t v_{p_1(\tau)}(\tau)\, d\tau\\
&=&x_{p_1(0)}(0)+ \int_0^t v_{p_1(\tau)}(\tau)\, d\tau,
\end{eqnarray*}
as desired.

This concludes the proof regarding $x_{p_1(t)}(t)$. 
The result for $x_{p_2(t)}(t)$ is obtained from
a symmetrical argument.
\end{proof}

\begin{lem}\label{lem:sort3}
Proposition \ref{prop:sort} holds when $n=3$.
\end{lem}
\begin{proof}
The main idea of the proof is to note that $\min_{i=1,2,3}x_i(t)=
\min\{x_1(t),\min\{x_2(t),x_3(t)\}\}$ and to use Lemma 
\ref{lem:sort2} twice.

Let $l(t) = 2$ if $x_2(t)\leq x_3(t)$, and let $l(t) =3$ otherwise. Let also
 $z_2(t) = x_{l(t)}(t)$ and $w_2(t) = v_{l(t)}(t)$. It follows
from Lemma \ref{lem:sort2}, applied to $x_2(t)$ and $x_3(t)$ that $z_2(t) = z_2(0) + \int_{0}^t
w_2(\tau)\, d\tau$. 

We let $z_1(t) = x_1(t)$ and $w_1(t) =
v_1(t)$. We then let $\lambda(t) =1$ if $y_1(t)\leq y_2(t)$, and $\lambda(t)=2$ otherwise. 
Using Lemma \ref{lem:sort2} once more, on $z_1(t)$ and $z_2(t)$, we have $z_{\lambda(t)}(t)=
z_{\lambda(0)}(0) + \int_{0}^t w_{\lambda(\tau)}(\tau)\,d\tau$.
Observe now that $\min_i x_i(t)=\min_i z_i(t)$, so that $p_1(t) =
1$ when $\lambda(t) =1$, and $p_1(t) = l(t)$ when $\lambda(t) =2$.
Therefore,
\begin{align*}
x_{p_1(t)}(t) &= z_{\lambda(t)}(t)\\ & = z_{\lambda(0)}(0) + \int_0^t
w_{\lambda(\tau)}(\tau)\, d\tau \\&= x_{p_1(0)}(0) + \int_0^t
v_{p_1(\tau)}(\tau)\, d\tau. 
\end{align*}
This proves the desired result
for $p_1(t)$. A symmetrical argument shows
the result for $p_3(t)$ as well. 

It remains to prove the result for 
$p_2(t)$. Observe that since $p$ is a permutation, we have 
$\sum_{i=1}^3  x_i(t)= \sum_{i=1}^3  x_{p_i(t)}(t)$ and
$\sum_{i=1}^3  v_i(t)= \sum_{i=1}^3  v_{p_i(t)}(t)$. Therefore, 
\begin{eqnarray*}
x_{p_2(t)}(t)&=& \prt{\sum_{i=1}^3 x_i(t)} - x_{p_1(t)}(t) -
x_{p_3(t)}(t) \\
&=& \prt{\sum_{i=1}^3 x_i(0)} - x_{p_1(0)}(0) - x_{p_3(0)}(0) \\ \hspace{-20pt}&+&
\int_{0}^t \prt{\prt{\sum_{i=1}^3 v_i(\tau)}- v_{p_1(\tau)}(\tau)
- v_{p_3(\tau)}(\tau) }d\tau\\ &=&x_{p_2(t)}(t) +
\int_{0}^tv_{p_2(\tau)}(\tau)d\tau.
\end{eqnarray*}
\end{proof}

We can now prove Proposition \ref{prop:sort}, using induction. 
\begin{proof}
Lemmas \ref{lem:sort2} and \ref{lem:sort3} establish the 
result for $n=2,3$. Suppose that the result holds for $n-1$, where $n\geq 4$;  
we will show that it also holds for $n$.

Let $q(t)$ be a permutation on $\{1,\dots,n-1\}$ such that, if $i<j$ then either $x_{q_i(t)}(t)< x_{q_j(t)}$ or $x_{q_i(t)}(t)=x_{q_j(t)}$ and $q_i(t) < q_j(t)$. 
(This corresponds to a sorting of the first $n-1$ components of $x(t)$, according to the same lexicographic rules used earlier to define $p(t)$.) 
It follows from our induction 
hypothesis that
$$
x_{q_{i}(t)}(t) = x_{q_{i}(0)}(0) + \int_{0}^t
v_{q_{i}(\tau)}(\tau)\, d\tau.$$ 
for $i=1,\dots,n-1.$

Let us now fix some $k\neq 1,n$. We will prove the desired result for
$x_{p_k(t)}$. 
Note that the sorted version of $x(t)$ (as captured by the coefficients $p_i(t)$) is obtained by inserting $x_n(t)$ into the sorted version of
the first $n-1$ components of $x(t)$ (as captured by the coefficients $q_i(t)$), at the appropriate position. In particular, 
there are only three possible values for $p_k(t)$, namely $q_k(t)$, $q_{k-1}(t)$, and $n$. In more detail, the value of $p_k(t)$
is determined as follows:
\begin{equation}\label{eq:pk_qk_qk-1}\begin{array}{rlllclll}
&\hspace{-5pt}x_n(t) &\hspace{-5pt}<\hspace{-5pt}& x_{q_{k-1}(t)}(t) & \Rightarrow  & p_k(t)  &\hspace{-5pt}=\hspace{-5pt}&
q_{k-1}(t),\\
x_{q_{k-1}(t)}(t) \leq &\hspace{-5pt}x_n(t) &\hspace{-5pt}<\hspace{-5pt}& x_{q_{k}(t)}(t) &\Rightarrow &
p_k(t) &\hspace{-5pt}=\hspace{-5pt}& n,
\\
x_{q_k(t)}(t)\leq &\hspace{-5pt}x_n(t) &\hspace{-5pt}\hspace{-5pt} &&\Rightarrow & p_k(t)  &\hspace{-5pt}=\hspace{-5pt}& q_k(t).
\end{array}\end{equation}
In order to focus on the three possible values of $p_k(t)$, we now define 
$z_1(t) = x_{q_{k-1}(t)}(t)$, $w_1(t) =
v_{q_{k-1}(t)}(t)$, $z_2(t) = x_{q_{k}(t)}(t)$, $w_2(t) =
v_{q_k(t)}(t)$, $z_3(t) = x_n(t)$, and $w_3(t) = v_n(t)$. Let then
$r(t)$ be a permutation on $\{1,2,3\}$ 
that sorts the components of $z(t)$ according to the same lexicographic rules used earlier for
$q$ and $p$. It follows
from (\ref{eq:pk_qk_qk-1}) and the definition of $z,w,r$, that
$z_{r_2(t)}= x_{p_k(t)}(t)$ and $w_{r_2(t)}(t) = v_{p_k(t)}(t)$.
Using Lemma \ref{lem:sort3}, we obtain
\begin{align*}
x_{p_k(t)}(t) &= z_{r_2(t)} \\& = z_{r_{2}(0)}(0) + \int_{0}^t
w_{r_{2}(\tau)}(\tau)\, d\tau\\ & = x_{p_k(0)}(0) + \int_{0}^t
v_{p_{k}(\tau)}(\tau)\, d\tau.
\end{align*}
This completes the proof of the result for
$k\not = 1,n$. The proof made use of the induction hypothesis together with Lemma \ref{lem:sort3}. The proof 
for the remaining cases ($k=1$ or $k=n$) is entirely similar, except that relies on Lemma
\ref{lem:sort2} instead of Lemma \ref{lem:sort3}. 
\end{proof}

%% file: cut_balance_for_arXiv.bbl
\begin{thebibliography}{10}

\bibitem{Ben-naimKrapivskyRedner:2003}
E.~Ben-Naim, P.L. Krapivsky, and S.~Redner.
\newblock Bifurcations and patterns in compromise processes.
\newblock {\em Physica D}, 183(3):190--204, 2003.

\bibitem{BertsekasTsitsiklis:1989}
D.P. Bertsekas and J.N. Tsitsiklis.
\newblock {\em Parallel and Distributed Computation: Numerical Methods.}
\newblock Prentice-Hall, Englewood Clifffs (NJ), USA, 1989.

\bibitem{BlondelHendrickxOlshevskyTsitsiklis:2005}
V.D. Blondel, J.M. Hendrickx, A.~Olshevsky, and J.N. Tsitsiklis.
\newblock Convergence in multiagent coordination, consensus, and flocking.
\newblock In {\em Proceedings of the 44th IEEE Conference on Decision and
  Control (CDC'2005)}, pp.\ 2996--3000, Seville, Spain, December 2005.

\bibitem{BlondelHendrickxTsitsiklis:2009_Krausemodel}
V.D. Blondel, J.M. Hendrickx, and J.N. Tsitsiklis.
\newblock On {K}rause's multi-agent consensus model with state-dependent
  connectivity.
\newblock {\em IEEE transactions on Automatic Control}, 54(11):2586--2597,
  2009.

\bibitem{BlondelHendrickxTsitsiklis:2009_APCT}
V.D. Blondel, J.M. Hendrickx, and J.N. Tsitsiklis.
\newblock Continuous-time average-preserving opinion dynamics with
  opinion-dependent communications.
\newblock {\em SIAM Journal on Control and Optimization}, 48(8):5214--5240,
  2010.

\bibitem{BoydGhoshPrabhakarShah:2005}
S.~Boyd, A.~Ghosh, B.~Prabhakar, and D.~Shah.
\newblock Gossip algorithms: Design, analysis and applications.
\newblock In {\em Proceedings of the 24th IEEE conference on Computer
  Communications (Infocom'2005)}, Vol.~3, pp.\ 1653--1664, Miami (FL), USA,
  March 2005.

\bibitem{CanutoFagnaniTilli:2008}
C.~Canuto, F.~Fagnani, and P.~Tilli.
\newblock A {E}ulerian approach to the analysis of rendez-vous algorithms.
\newblock In {\em Proceedings of the 17th {IFAC} {W}orld {C}ongress ({IFAC'08})},
  pp.\ 9039--9044, July 2008.

\bibitem{CastellanoFortunatoLoreto:2009}
C.~Castellano, S.~Fortunato, and V.~Loreto.
\newblock {Statistical physics of social dynamics}.
\newblock {\em Reviews of Modern Physics}, 81(2):591--646, 2009.

\bibitem{Chazelle:2009}
B.~Chazelle.
\newblock {The convergence of bird flocking}.
\newblock {\em arXiv:0905.4241v1 [cs.CG]}, 2009.

\bibitem{DeffuantNeauAmblardWeisbuch:2000}
G.~Deffuant, D.~Neau, F.~Amblard, and G.~Weisbuch.
\newblock Mixing beliefs among interacting agents.
\newblock {\em Advances in Complex Systems}, 3:87--98, 2000.

\bibitem{FortunatoLatoraPluchinoRapisarda:2005}
S.~Fortunato, V.~Latora, A.~Pluchino, and Rapisarda R.
\newblock Vector opinion dynamics in a bounded confidence consensus model.
\newblock {\em International Journal of Modern Physics C}, 16:1535--1551,
  2005.

\bibitem{HegselmannKrause:2002}
R.~Hegselmann and U.~Krause.
\newblock Opinion dynamics and bounded confidence models, analysis, and
  simulations.
\newblock {\em Journal of Artificial Societies and Social Simulation},
  5(3), 2002. \url{http://jasss.soc.surrey.ac.uk/9/1/8.html}

\bibitem{Hendrickx:2008phdthesis}
J.M. Hendrickx.
\newblock {\em Graphs and Networks for the Analysis of Autonomous Agent
  Systems}.
\newblock PhD thesis, Universit{\'e} catholique de Louvain, 2008.
%\url{http://www.inma.ucl.ac.be/~hendrickx/availablepublications/Thesis_Julien_Hendrickx.pdf}
\url{http://perso.uclouvain.be/julien.hendrickx/availablepublications/Thesis_Julien_Hendrickx.pdf}

\bibitem{HendrickxBlondel:2006_MTNS}
J.M. Hendrickx and V.D. Blondel.
\newblock Convergence of different linear and non-linear {V}icsek models.
\newblock In {\em Proceedings of the 17th International Symposium on
  Mathematical Theory of Networks and Systems (MTNS'2006)}, pp.\ 1229--1240,
  Kyoto, Japan, July 2006.


\bibitem{HendrickxTsitsiklis_cutbalanced:2011_cdc}
J.M. Hendrickx and J.N. Tsitsiklis
\newblock A new condition for convergence in continuous-time consensus seeking systems.
\newblock In {\em Proceedings of the 50th IEEE Conference on Decision and Control (CDC 2011)}, Orlando (FL, USA), December 2011.


\bibitem{HendrickxTsitsiklis_cutbalanced:2013}
J.M. Hendrickx and J.N. Tsitsiklis
\newblock Convergence of type-symmetric and cut-balanced consensus seeking systems.
\newblock  {\em IEEE Transactions on Automatic Control},
58(1):214-218, 2013.






\bibitem{JadbabaieLinMorse:2003}
A.~Jadbabaie, J.~Lin, and A.~S. Morse.
\newblock Coordination of groups of mobile autonomous agents using nearest
  neighbor rules.
\newblock {\em IEEE Transactions on Automatic Control}, 48(6):988--1001, 2003.


\bibitem{LiWang:2004}
S.~Li and H.~Wang.
\newblock Multi-agent coordination using nearest neighbor rules: a revisit to
  {V}icsek model.
\newblock {\em arXiv:cs/0407021v2 [cs.MA]}, 2004.

\bibitem{LinBrouckeFrancis:2004}
Z.~Lin, M.~Broucke, and B.~Francis.
\newblock {Local control strategies for groups of mobile autonomous agents}.
\newblock {\em IEEE Transactions on Automatic Control}, 49(4):622--629, 2004.

\bibitem{Lorenz:2003diplomathesis}
J.~Lorenz.
\newblock {\em Mehrdimensionale {M}einungsdynamik bei wechselndem {V}ertrauen}.
\newblock PhD thesis, Universitat Bremen, 2003.
  \url{http://elib.suub.uni-bremen.de/dipl/docs/00000056.pdf} 

\bibitem{Lorenz:2005}
J.~Lorenz.
\newblock A stabilization theorem for continuous opinion dynamics.
\newblock {\em Physica A}, 355(1):217--223, 2005.

\bibitem{Lorenz:2007}
J.~Lorenz.
\newblock Continuous opinion dynamics under bounded confidence: A survey.
\newblock {\em International Journal of Modern Physics C}, 18(12):1819--1838,
  2007.

\bibitem{MateiBaras:2009}
I.~Matei and J.S. Baras.
\newblock Convergence results for the linear consensus problem under Markovian
  random graphs.
\newblock ISR Technical report 2009-18, University of Maryland, 2009.
  \url{http://drum.lib.umd.edu/bitstream/1903/9693/5/IMateiJBaras.pdf}

\bibitem{MateiMartinsBaras:2008}
I.~Matei, N.~Martins, and J.S. Baras.
\newblock Almost sure convergence to consensus in Markovian random graphs.
\newblock {\em 47th IEEE Conference
  on Decision and Control,} 2008, pp.\ 3535--3540. 

\bibitem{MirtabatabaeiBullo:2011}
A.~Mirtabatabaei and F.~Bullo.
\newblock On opinion dynamics in heterogeneous networks.
\newblock {\em Submitted,} 2010. \url{http://motion.me.ucsb.edu/pdf/2010s-mb.pdf}

\bibitem{Moreau:2004}
L.~Moreau.
\newblock Stability of continuous-time distributed consensus algorithms.
\newblock In {\em Proceedings of the 43st IEEE Conference on Decision and
  Control}, Vol.~4, pp. 3998--4003, Paradise Island, Bahamas,
  December 2004.

\bibitem{Moreau:2005}
L.~Moreau.
\newblock Stability of multiagent systems with time-dependent communication
  links.
\newblock {\em IEEE Transactions on Automatic Control}, 50(2):169--182, 2005.

\bibitem{OlfatiSaberFaxMurray:2006}
R.~Olfati-Saber, J.A. Fax, and R.M. Murray.
\newblock Consensus and cooperation in networked multi-agent systems.
\newblock {\em Proceedings of the IEEE}, 95(1):215--233, 2007.

\bibitem{OlfatiSaberMurray:2004}
R.~Olfati-Saber and R.M. Murray.
\newblock Consensus problems in networks of agents with switching topology and
  time-delays.
\newblock {\em IEEE Transactions on Automatic Control}, 49(9):1520--1533,
  2004.

\bibitem{RenBeard:2005}
W.~Ren and R.W. Beard.
\newblock Consensus seeking in multiagent systems under dynamically changing
  interaction topologies.
\newblock {\em IEEE Transactions on Automatic Control}, 50:655--661, 2005.

\bibitem{RenBeardAtkins:2007}
W.~Ren, R.W. Beard, and E.M. Atkins.
\newblock Information consensus in multivehicle cooperative control.
\newblock {\em IEEE Control and Systems Magazine}, 27(2):71--82, 2007.


\bibitem{Sontag:1998}
E.D.~Sontag, 
\newblock {\em Mathematical Control Theory: Deterministic Finite Dimensional Systems.}
\newblock Second Edition, Springer, New York, 1998.



\bibitem{TahbazJadbabaie:2010}
A.~Tahbaz-Salehi and A.~Jadbabaie.
\newblock {Consensus over ergodic stationary graph processes}.
\newblock {\em IEEE Transactions on Automatic Control,} 55(1):225--230, 2010.

\bibitem{TouriNedic:2010}
B.~Touri and A.~Nedic.
\newblock {When Infinite Flow is Sufficient for Ergodicity}.
\newblock In {\em Proceedings of the 49th IEEE Conference on Decision and
  Control (CDC'2010)}, Atlanta (GA, USA), 2010.
  \url{https://netfiles.uiuc.edu/touri1/www/Papers/ergodicityCDC.pdf} 

\bibitem{Tsitsiklis:84phdthesis}
J.N. Tsitsiklis.
\newblock {\em Problems in decentralized decision making and computation}.
\newblock PhD thesis, Dept. of Electrical Engineering and Computer Science,
  Massachusetts Institute of Technology, 1984.
  \url{http://web.mit.edu/jnt/www/PhD-84-jnt.pdf}

\bibitem{VicsekCzirolBenjacobCohenSchchet:1995}
T.~Vicsek, A.~Czirok, I.~Ben~Jacob, I.~Cohen, and O.~Schochet.
\newblock Novel type of phase transitions in a system of self-driven particles.
\newblock {\em Physical Review Letters}, 75:1226--1229, 1995.

\bibitem{XiaoWang:2008}
F.~Xiao and L.~Wang.
\newblock {Asynchronous consensus in continuous-time multi-agent systems with
  switching topology and time-varying delays}.
\newblock {\em IEEE Transactions on Automatic Control}, 53(8):1804--1816, 2008.

\bibitem{XiaoBoydLall:2005}
L.~Xiao, S.~Boyd, and S.~Lall.
\newblock {A scheme for robust distributed sensor fusion based on average
  consensus}.
\newblock In {\em Fourth International Symposium on Information Processing in Sensor Networks (IPSN 2005)}, 63--70, 2005.

\end{thebibliography}
